\documentclass[aps,pra,floatfix,superscriptaddress,11pt,nofootinbib,longbibliography]{revtex4-2}
\pdfoutput=1

\usepackage[bookmarksnumbered,hypertexnames=false]{hyperref}
\usepackage{bbm,braket,microtype,mathrsfs,amsmath,amssymb,color,amsthm,graphicx}
\usepackage[T1]{fontenc}
\usepackage[USenglish]{babel}
\usepackage{times}
\usepackage{verbatim}
\usepackage[capitalize]{cleveref}
\usepackage{algorithmic}
\usepackage{bm}
\usepackage{makecell}
\usepackage{longtable}
\usepackage{multirow}
\usepackage{booktabs}
\newcommand{\tabitem}{~~\llap{\textbullet}~~}
\newtheorem{thm}{Theorem}\crefname{thm}{Theorem}{Theorems}
\newtheorem{lem}[thm]{Lemma}\crefname{lem}{Lemma}{Lemmas}
\newtheorem{rem}[thm]{Remark}\crefname{rem}{Remark}{Remarks}
\newtheorem{cor}[thm]{Corollary}\crefname{cor}{Corollary}{Corollaries}
\newtheorem{dfn}[thm]{Definition}\crefname{dfn}{Definition}{Definitions}

\DeclareMathOperator{\RS}{Row Span}
\DeclareMathOperator{\RM}{RM}
\DeclareMathOperator{\rank}{rank}

\DeclareMathOperator{\diag}{diag}
\renewcommand{\vec}{\mathbf}

\newcommand{\ZZ}{\mathbb Z}

\newcommand{\FF}{\mathbb F}

\newcommand{\calH}{\mathcal H}
\newcommand{\calG}{\mathcal G}
\newcommand{\calX}{\mathcal X}
\newcommand{\calZ}{\mathcal Z}

\newcommand{\utri}{unital triorthogonal }

\newcommand{\tri}{triorthogonal }
\newcommand{\Tri}{Triorthogonal }

\renewcommand{\i}{{\mathrm{i}}}

\newcommand{\ssection}[1]{\smallskip\phantomsection\addcontentsline{toc}{section}{#1}\textit{#1.---}}

\allowdisplaybreaks[4]
\begin{document}

\title{Classification of Small \Tri Codes}
\author{Sepehr Nezami}
\affiliation{Institute for Quantum Information and Matter and Walter Burke Institute for Theoretical Physics, California Institute of Technology, Pasadena, California, USA}
\author{Jeongwan Haah}
\affiliation{Microsoft Quantum, Redmond, Washington, USA}
\begin{abstract}
Triorthogonal codes are a class of quantum error correcting codes used in magic state distillation protocols.
We classify all triorthogonal codes with $n+k\leq 38$, 
where $n$ is the number of physical qubits and $k$ is the number of logical qubits of the code. 
We find $38$ distinguished triorthogonal subspaces and 
show that every triorthogonal code with $n+k\leq 38$ 
descends from one of these subspaces through elementary operations such as puncturing and deleting qubits.
Specifically,
we associate each triorthogonal code with a Reed--Muller polynomial of weight $n+k$, 
and classify the Reed--Muller polynomials of low weight using the results of 
Kasami, Tokura, and Azumi~\cite{kasami1970weight,kasami1976weight} 
and an extensive computerized search.
In an appendix independent of the main text,
we improve a magic state distillation protocol by reducing the time variance due to stochastic Clifford corrections.
\end{abstract}

\maketitle

\section{Introduction}

A magic state is a state on one or more qubits 
with which Clifford gates and Pauli measurements complete quantum universality~\cite{knill2004,bravyi2005universal}.
Clifford operations can be implemented fault-tolerantly using Pauli stabilizer codes,
and magic states of high fidelity can be distilled using Clifford operations.
This way of achieving fault-tolerant quantum universality underlies leading proposals 
for quantum computers at scale~\cite{karzig2017scalable,Chamberland2020,Beverland2021,Bombin2021}.
However, the fault tolerance for nonClifford operations via magic state distillation 
is estimated more costly than that for Clifford operations
and hence there has been many proposals to reduce the cost;
see~\cite{Litinski2019,Chamberland2020} and references therein.

A broad class of magic state distillation protocols~\cite{knill2004,bravyi2005universal,jones2013multilevel,paetznick2013,haah2017magic,haah2017magic2,campbell2017unifying,haah2018codes}
uses so-called triorthogonal codes~\cite{bravyi2012magic}.
These are a class of CSS codes that admit transversal gates 
at one level higher in the Clifford hierarchy 
than Clifford gates, and are specified by certain cubic polynomial equations.
Even if a magic state distillation protocol 
does not nominally involve a triorthogonal code,
many protocols correspond to triorthogonal codes after some manipulation~\cite{haah2018towers}.
Given a triorthogonal code,
there are various ways to implement a magic state distillation protocol~\cite{bravyi2012magic,haah2018codes,Litinski2019}.
In all cases, if a protocol corresponds to a triorthogonal code,
the basic parameters of the triorthogonal code 
(the encoding rate and code distance)
have direct consequences in the performance of the protocol.
Hence, it is natural to seek optimal triorthogonal codes as an abstract CSS code.
A few infinite classes of triorthogonal codes are known to date~\cite{bravyi2012magic,haah2018codes,hastings2018distillation,haah2018towers},
but extremal codes (those of the best encoding rate given a code distance) 
are still poorly understood.
One could instead ask for a complete table of small triothogonal codes,
with which one would be able to optimize Clifford circuits 
implementing magic state distillation protocols.

In this paper, we give results towards the classifcation of triorthogonal codes,
which are useful at least for short length codes.
We associate an indicator polynomial to any triorthogonal matrix,
which is naturally identified with a codeword of a Reed--Muller code.
Specifically, we regard a triorthogonal matrix as a collection of column vectors,
which is then identified with the support of an indicator function.
This approach allows us to use existing classification results
on Reed--Muller codewords of small weights~\cite{kasami1970weight,kasami1976weight} 
and new computerized searches to classify all triorthogonal codes up to certain size.

Generalizing the puncturing procedure of~\cite{haah2018codes},
we focus on triorthogonal spaces rather than triorthogonal codes,
where the latter is obtained from the former by choosing a set of coordinates.
We have run a computer-assisted exhaustive search over the choices of these sets of coordinates, and report all Reed--Muller polynomials corresponding to the triorthogonal spaces and the distances of respective codes in~\cref{tab:2}. See~\cref{fig:classification} as well.

Notable new examples from our search include
codes of parameters $[[28,2,3]]$ and $[[35,3,3]]$. See~\cref{eq:gen35} for the generator matrix of the $[[35,3,3]]$ code. 
\begin{figure}[t]
    \centering
    \includegraphics[width=16.3cm]{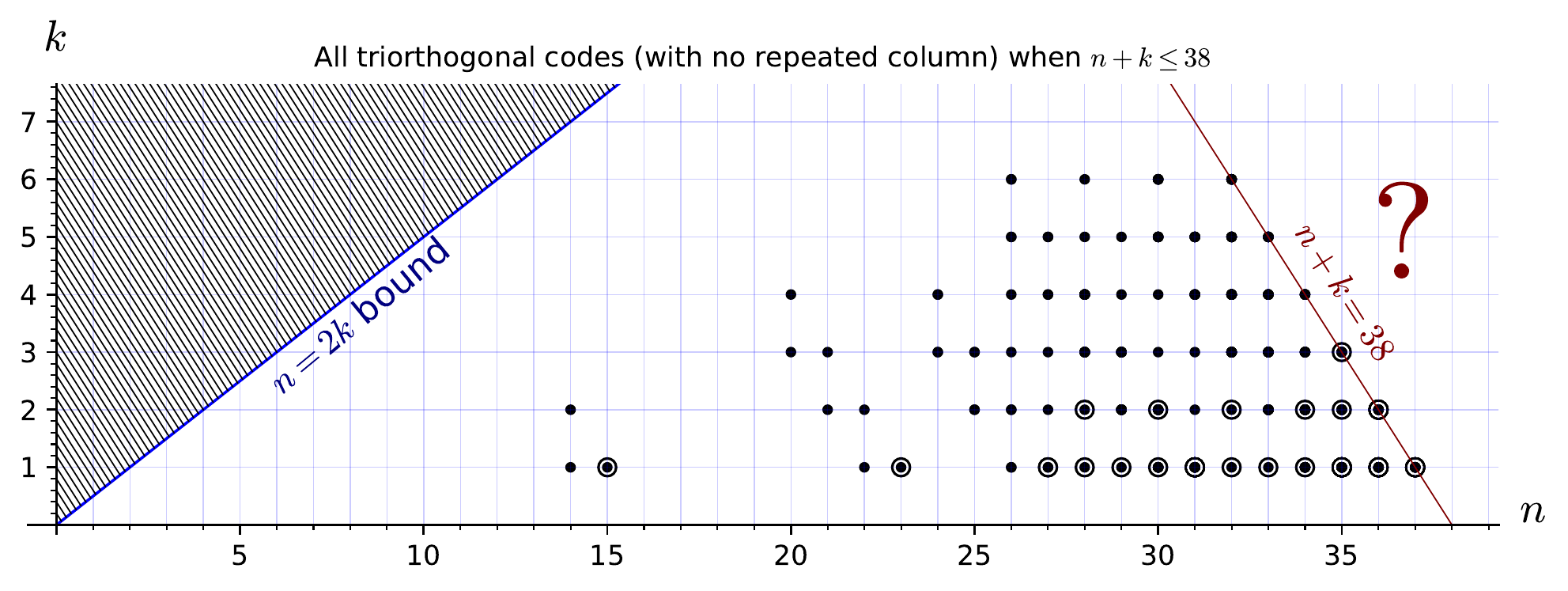}
    \caption{List of all possible pairs $(n,k)$ with $n+k\leq 38$ such that a triorthogonal code of parameters $[[n,k,d]]$ with $d\geq 2$ exist. Small solid dots indicate the cases where the \tri codes' maximum attainable distance is $2$, while the dots enclosed in a circle correspond to the cases where the maximum achievable distance is $3$. Note that there is no \tri code of distance $4$ when $n+k\leq 38$. We know that the hatched region above the $n=2k$ line contains no \tri code as a result of~\cref{lem:bound}. The region on the right of $n+k=38$ in the above figure is unexplored in our investigation, and our result is exhaustive only in the region to the left of (and including) the line $n+k=38$. We have excluded non-unital triorthogonal codes (defined in~\cref{sec:unitalcode}) and codes with repeated columns, as they can easily be constructed from the above codes and will not have better parameters (see \cref{sec:unitalcode} and~\cite{bravyi2012magic}).}
    \label{fig:classification}
\end{figure}

\begin{equation}\label{eq:gen35}
\setcounter{MaxMatrixCols}{35 }
{\scriptsize
\begin{bmatrix}
 & & & & & & & &\bm{1}&\bm{1}&\bm{1}& & & & & & &\bm{1}&\bm{1}&\bm{1}&\bm{1}&\bm{1}&\bm{1}& & & & & & &\bm{1}&\bm{1}&\bm{1}&\bm{1}&\bm{1}&\bm{1}\\
 & & & & & & & & & & &\bm{1}&\bm{1}&\bm{1}& & & & & & & & & &\bm{1}&\bm{1}&\bm{1}&\bm{1}&\bm{1}&\bm{1}&\bm{1}&\bm{1}&\bm{1}&\bm{1}&\bm{1}&\bm{1}\\
 & & & & & & & & & & & & & &\bm{1}&\bm{1}&\bm{1}&\bm{1}&\bm{1}&\bm{1}&\bm{1}&\bm{1}&\bm{1}&\bm{1}&\bm{1}&\bm{1}&\bm{1}&\bm{1}&\bm{1}&\bm{1}&\bm{1}&\bm{1}&\bm{1}&\bm{1}&\bm{1}\\\hline
\bm{1}& & & & & &\bm{1}& &\bm{1}& &\bm{1}& & & &\bm{1}& &\bm{1}&\bm{1}& & &\bm{1}&\bm{1}&\bm{1}& &\bm{1}&\bm{1}& &\bm{1}&\bm{1}&\bm{1}& & &\bm{1}& & \\
 &\bm{1}& & & & &\bm{1}& &\bm{1}&\bm{1}& & & & &\bm{1}&\bm{1}& & &\bm{1}& &\bm{1}&\bm{1}&\bm{1}&\bm{1}& &\bm{1}&\bm{1}& &\bm{1}& &\bm{1}& & &\bm{1}& \\
 & &\bm{1}& & & &\bm{1}& & &\bm{1}&\bm{1}& & & & &\bm{1}&\bm{1}& & &\bm{1}&\bm{1}&\bm{1}&\bm{1}&\bm{1}&\bm{1}& &\bm{1}&\bm{1}& & & &\bm{1}& & &\bm{1}\\
 & & &\bm{1}& & & &\bm{1}& & & &\bm{1}& &\bm{1}&\bm{1}& &\bm{1}& &\bm{1}&\bm{1}& &\bm{1}&\bm{1}&\bm{1}&\bm{1}&\bm{1}&\bm{1}& & & & &\bm{1}& &\bm{1}& \\
 & & & &\bm{1}& & &\bm{1}& & & &\bm{1}&\bm{1}& &\bm{1}&\bm{1}& &\bm{1}& &\bm{1}&\bm{1}& &\bm{1}&\bm{1}&\bm{1}&\bm{1}& &\bm{1}& & & & &\bm{1}& &\bm{1}\\
 & & & & &\bm{1}& &\bm{1}& & & & &\bm{1}&\bm{1}& &\bm{1}&\bm{1}&\bm{1}&\bm{1}& &\bm{1}&\bm{1}& &\bm{1}&\bm{1}&\bm{1}& & &\bm{1}&\bm{1}&\bm{1}& & & & \end{bmatrix}
}
\end{equation}
Furthermore, we show that there is no code of distance larger than $3$ in this regime of parameters. 

\section{Unital triorthogonal subspaces and descendant codes} \label{sec:unital}

We start with the following definition:
\begin{dfn}[Unital triorthogonal subspaces]\label{dfn:unital}
A subspace $\calH \subseteq \FF_2^c$ of dimension $r$ is \emph{triorthogonal} 
if for any three vectors $\vec u, \vec v, \vec w \in \calH$ we have $\sum_{i=1}^c \vec u^i \vec v^i \vec w^i = 0 \bmod 2$
where $\vec u^i$ denotes the $i$-th component of the vector~$\vec u$.
If~$\calH$ contains all-one vector~$\vec 1_c$, (i.e., $\vec 1_c^i = 1$ for all~$i$),
then $\calH$ is called \emph{unital}.
\end{dfn}
In the definition, the three vectors need not be distinct.
Hence, any triorthogonal subspace is always self-orthogonal.
If there is a unital triorthogonal subspace in $\FF_2^c$, then $c$ must be even because $\vec 1_c$ is orthogonal to itself.
As always in coding theory,
a permutation of coordinates is considered an equivalence transformation.
If two subspaces $\calH$ and $\calH'$ are the same up to permutations of components,
we will write
\begin{align}
    \calH \cong \calH'
\end{align}
and say that they are \emph{isomorphic}. 
We generally present the \tri subspaces as the row space of a matrix, i.e., 
\begin{align}\label{eq:genmat}
\calH = \RS (H),
\end{align}
where $H$ is called the~\emph{generator matrix} of $\calH$. Lastly, similar to the subspaces, we say two matrices are isomorphic and present it by
\begin{align}\label{eq:geniso}
    H \cong H'
\end{align}
if they can be converted to each other by a permutation of their columns.

\subsection{Descendant codes}

Let $\calH \subseteq \FF_2^c$ be a unital triorthogonal subspace.
Suppose we are given a set $P \subset \{1,2,\ldots,c\}$ of $p = |P|$ coordinate labels
such that the restriction of $\calH$ on these coordinates has dimension $p$:
to be more clear, we define a restriction linear map 
\begin{align}
\Pi_P : e_j \mapsto \begin{cases} e_j & \text{if } j \in P \\ 0, & \text{otherwise} \end{cases}
\end{align}
for all $j =1,2,\ldots,c$
where $e_j$ is the standard basis vector of $\FF_2^c$ with sole~$1$ at the $j$-th position.
Although $\Pi_P$ is a map from $\FF_2^c$ to itself,
the codomain may be regarded as $\FF_2^p$.
This amounts to the puncturing procedure for classical codes.
Under this convention,
$\Pi_P \calH$ is a subspace of $\FF_2^{p}$. 
For the rest of this manuscript, we always consider $P$ such that the codomain of $\Pi_P$ coincides its image:
\begin{align}
\Pi_P \calH = \FF_2^p. \label{eq:fullrank}
\end{align}
From $\calH$, we can define a quantum CSS code with $n$ physical qubits and $k$ logical qubits by the following procedure~\cite{haah2018codes}.
\begin{enumerate}
    \item[Even.] ($\boldsymbol{n+k=0\bmod 2}$)
        Put $k=p$ and $n = c - p$ where $p < c/2$ so that $n > k$.
        Choose any basis for $\calH$ and put it in the rows of an $r$-by-$c$ matrix $G'$.
        Bring the columns of $G'$ corresponding to $P$ to the left by a column permutation,
        and put the resulting matrix in the reduced row echelon form:
        \begin{align}
            \calH \cong \RS \left[\begin{array}{c|c}
                                    I_k&G_1  \\\hline
                                    0&G_0 
                                \end{array}  \right] \label{eq:evenDes}
        \end{align}
        By assumption \cref{eq:fullrank}, there has to be an $k$-dimensional identity matrix $I_k$ on the top left.
        
    \item[Odd.] ($\boldsymbol{n+k=1\bmod 2}$)
        Put $k = p - 1$ and $n = c - p$ where $p < (c+1)/2$ so that $n > k$.
        Choose a basis of $\calH$ by extending $\{ \vec 1 \}$
        and put the basis in the rows of an $r$-by-$c$ matrix $G'$.
        The first row of $G'$ is $\vec 1$.
        Bring the columns of $G'$ corresponding to $P$ to the left by a column permutation,
        and put all the rows but the first into the reduced row echelon form:
        \begin{align}
        \calH \cong \RS 
            \left[\begin{array}{c|c|c}
                \multicolumn{3}{c}{\vec 1_{n+k+1}}    \\\hline
                        0 & I_k&G_1 \\\hline
                        0 & 0  &G_0
                \end{array}  \right] \label{eq:oddDes}
        \end{align}
        By assumption \cref{eq:fullrank}, the top left submatrix must have rank $k+1$,
        and the Gaussian elimination reveals the displayed $k$-dimensional identity matrix.
\end{enumerate}

It is straightforward to check that the submatrix $G$ consisting of $G_0$ and $G_1$ 
satisfies the following conditions~\cite{bravyi2012magic}:
\begin{align}
    \sum_{j} G^{a,j} G^{b,j} =0 \bmod 2\quad \text{for all } a<b,\qquad \text{ and}\label{eq:orth}\\
    \sum_{j} G^{a,j} G^{b,j} G^{c,j} =0 \bmod 2\quad \text{for all }a<b<c,\label{eq:triorth}
\end{align}
where $G^{a,j}$ is the matrix element of $G$ at the $a$th row and the $j$th column.
In addition, the rows of $G_0$ have even weight, and those of $G_1$ odd.

\begin{dfn}
A binary matrix $G$ is \emph{triorthogonal} if it satisfies both \cref{eq:orth,eq:triorth}.
\end{dfn}

Now, let $\calX$ be the row span of submatrix $G_0$,
and $\calZ$ be the orthogonal complement of the rows of $G_0$ and $G_1$.
We have $\dim \calX = r - p$ and $\dim \calZ = n - k - r + p$.
Hence, the CSS code defined by $X$-stabilizers corresponding to $\calX$ and $Z$-stabilizers $\calZ$,
encodes $k$ logical qubits into $n$ qubits.
The rows of $G_1$ are orthogonal to each other and to the rows of $G_0$;
this is inherited from the self-orthogonality of $\calH$.
We make a specific choice of $X$ logical operators by declaring that each row of $G_1$ corresponds to a $X$ logical operator.
We also choose the $Z$ logical operators
by declaring that each row of $G_1$ corresponds to a $Z$ logical operator.
This choice of $X$ and $Z$ logical operators
determines a decomposition of the code space into logical qubits.
This is a generalization of the puncturing process of~\cite{haah2018codes};
the odd descendants have not been considered before.
The relevant distance for magic state distillation is the $Z$ distance,
the minimum of the weight of any nontrivial $Z$ logical operator:
\begin{align}
    d_Z = \min_{z \in G_0^\perp \setminus G^\perp} |z| .
\end{align}

Although defined in terms of a basis of $\calH$, the descendant triorthogonal codes are independent of the basis.
\begin{lem}
Any even descendant triorthogonal code and its choice of $X$ logical operators depends only on $P$ as a set,
not on the ordering of coordinates within $P$.
Any odd descendant triorthogonal code and its choice of $X$ logical operators depends only on a pair $(P, j \in P)$,
not on the ordering of coordinates within $P \setminus \{j\}$.
\end{lem}
\begin{proof}
(Even case)
    Let $Q = \{1,2,\ldots,c\} \setminus P$ be the complementary coordinate set.
    The subspace $\calX$ is the kernel of $\Pi_P$
    and $\calZ$ is the orthogonal complement of $\Pi_Q \calH$ within $\FF_2^n$.
    This shows that the stabilizer group only depends on $P$, not on the ordering within $P$.
    A different ordering within $P$ corresponds to a permutation on the coordinates of $P$,
    which is represented by a permutation matrix multiplied on the right of $G'$, the matrix in \cref{eq:evenDes},
    that acts nontrivially only on the first $k$ columns.
    This permutation can be compensated by its inverse acting on the left of $G'$,
    permuting first $k$ rows of $G'$.
    This row permutation leads to a different matrix $G_1$ but the overall matrix remains in a row echelon form that is not necessarily reduced.
    Applying (reverse) Gauss elimination, we see that $G_0$ part remains intact,
    and $G_1$ part will be modified by $G_0$.
    As logical operators, this modification is simply multiplications by $X$ stabilizers.
    
(Odd case)
    The distinguished coordinate label $j$ determines the first column of the matrix in \cref{eq:oddDes}.
    The first column and the first row of $\vec 1$ are the only difference from the even case.
    But the vector $\vec 1$ is a permutation invariant vector, so the argument for the even case applies here.
\end{proof}

\subsection{Triorthogonal matrices to unital triorthogonal spaces}\label{sec:unitalcode}

\begin{dfn}
A triorthogonal matrix or code is \emph{unital}
if it is obtained by one of the descending procedures from a unital triorthogonal subspace.
\end{dfn}

A triorthogonal matrix might not be unital.
However, considering unital ones is not constraining as we show.
First, if $n+k$ is odd, $G$ is already an odd descendant of a unital triorthogonal space:
\begin{align}
\label{eq:padding_odd}
G=\left[\begin{array}{c}
     G_1  \\\hline
     G_0
\end{array}   \right]
\xleftarrow{\text{odd descending}}
\calH =\RS
\left[\begin{array}{c|c|c}
         \multicolumn{3}{c}{\vec 1_{n+k+1}}    \\\hline
     0 & I_k&G_1 \\\hline
     0 & 0  &G_0
\end{array}  \right].
\end{align}
Second, if $n+k$ is even, then a vector $\vec v = \vec 1_{n} + G_1^1 + \cdots + G_1^k$ has even weight
because the mod 2 weight of the sum of vectors is just the mod 2 sum of all weights of all the vectors,
and each $G_1^j$ with $j=1,\ldots,k$ has odd weight.
Hence, $\vec v$ can be adjoined to the $X$ stabilizer group $\calX$ to form a new code
unless it is already in $\calX$.
This amounts to enlarging $G_0$ with an additional row $\vec v$.
The new vector addition does not violate the triorthogonality as one can easily check.
It is straightforward to see that this addition can only increase the $Z$ distance of the code,
without changing the number of logical or physical qubits.
Note that this addition may decrease the $X$ distance since the set of all representatives of
a $X$ logical operator is enlarged;
however, we only care about $Z$ distances in this paper.
Now, this new code is unital,
\begin{align}
G=\left[\begin{array}{c}
     G_1  \\\hline
     G_0
\end{array}   \right] \xleftarrow{\text{even descending}} \calH =\RS \left[\begin{array}{c|c}
      I_k&G_1  \\\hline
     0&G_0 
\end{array}  \right] \label{eq:padding_even}
\end{align}
as the sum of $\vec v$ and the first $k$ rows of the matrix on the right-hand side
is the all-one vector $\vec 1_{k+n}$.
For this reason, we only consider unital triorthogonal codes in this paper.

\section{Connection to the Reed--Muller codewords}

\subsection{Review of binary Reed--Muller codes}
Define a set $\RM(s,m)$ to be the collection of all polynomials in $m$ binary variables of degree at most $s$, with coefficients in $\FF_2$:
\begin{align}
\RM(s,m) := \left\{p\in \FF_2[x_1,\cdots,x_m]/( x_i^2-x_i) ~ :~ \deg  p \leq s \right\}.  
\end{align}
Strictly speaking, since we are modding out the polynomial ring by the ideal $(x_i^2 - x_i)$,
the degree is not the usual one;
for us, the degree of a monomial is simply the number of variables with nonzero exponent in the monomial,
and the degree of a polynomial is the maximum degree of all nonzero monomials in the polynomial.
In this convention,
\begin{align}
    \deg (fg) \le \deg f + \deg g,
\end{align}
which may be strict even if $f \neq 0$ and $g \neq 0$.
For example, take $f = x_1$ and $g = x_1$.
An element of $\RM(s,m)$ is identified with its value list:
\begin{align}
    p \in \RM(s,m) 
    \Longleftrightarrow 
    ( p(x) \in \FF_2 ~:~ x \in \FF_2^m )
\end{align}
The collection of all value lists is the Reed--Muller code.
The minimum distance of the Reed--Muller code is~$2^{m-s}$.
The dual (orthogonal complement) of $\RM(s,m)$ is
$\RM(s,m)^\perp = \RM(m-s-1,m)$.

\subsection{Indicator polynomials}

Every binary function is uniquely specified by its support, which is the set of all inputs that evaluate to~$1$,
and any subset of a $\FF_2$-vector space specifies an \emph{indicator} function that assumes~$1$ precisely on the subset.
For a matrix $H$ that is $r$-by-$c$ 
where the first row is the all-one vector
and no columns are repeated,
we associate a unique polynomial $f$ by the following rule.
\begin{align}
    f &\Longleftrightarrow H \nonumber\\
  f(x_1,\ldots,x_{r-1}) = 1 &\Longleftrightarrow (1, x_1,\ldots,x_{r-1}) \text{ is a column of } H  \label{eq:indicator}
\end{align}
By slight abuse of language,
we call $f$ the indicator polynomial of $H$.
We will only consider matrices whose first row is all-one,
so this will not cause any confusion.

Hence, any generator matrix for a unital triorthogonal space
gives an indicator polynomial.
The number of columns in a generator matrix
is equal to the Hamming weight of the indicator polynomial viewed as a Reed--Muller codeword.

We can now characterize indicator polynomials for unital triorthogonal spaces.
\begin{lem}
Let $H$ be an $r$-by-$c$ binary matrix with no columns repeated and the first row being the all one vector,
and let $f \in \FF_2[x_1,\ldots,x_{r-1}]/(x_i^2 - x_i)$ be its indicator polynomial (in the sense of~\cref{eq:indicator}).
The row span of $H$ is unital triorthogonal
if and only if
$\deg f \le r - 5$.
\end{lem}

It follows that the smallest unital triorthogonal space
requires $r-1=4$ variables and with indicator polynomial $f = 1$,
which is precisely $\RM(1,4)$ on $16$ bits.

\begin{proof}
For any $a=1,2,\ldots,r-1$,
the $a$-th row of $H$ is the value list of a function~$x_a f$
over the support of~$f$.
Since the weight of any row of $H$ is even, 
we have $\sum_x x_a f = 0 \bmod 2$ for any $a$ and also $\sum_x f = 0 \bmod 2$.
Moreover, the componentwise product of two rows $a$ and $b$ is the value list of $x_a f(x)  \cdot x_b f(x) = x_a x_b f(x)$.
Hence, the self-orthogonality is equivalent to $\sum_x x_a x_b f = 0 \bmod 2$.
Similarly, for a triple $a,b,c$ of rows, the triple overlap is zero if and only if $\sum_x x_a x_b x_c f = 0 \bmod 2$.

Therefore, the indicator polynomial of $H$ should be orthogonal to all polynomials with degree $\leq 3$, and we have 
\begin{align}
f\in  \RM(3,r-1)^\perp =  \RM(r-5,r-1).
\end{align}
The converse is obvious.
\end{proof}

There is no canonical choice of a generator matrix given a triorthogonal space as one may apply row operations on the generator matrix without changing its row span.
But the row operations give all possible generator matrices (up to column permutations),
so we only have to consider how an indicator polynomial transforms upon a row operation. 
Let $x \mapsto Lx+\ell$ be an invertible affine transformation on $\FF_2^{r-1}$. 
It is easy to see that for any $v \in \FF_2^{r-1}$,
\begin{align}
    \begin{pmatrix}1\\ Lv+\ell\end{pmatrix} 
    = \begin{pmatrix} 1 & 0 \\ \ell & L \end{pmatrix} \begin{pmatrix}1 \\ v\end{pmatrix} \text{ is a column of }H 
    &\Longleftrightarrow
    f(L v+\ell) = 1 \nonumber \\
    &\Longleftrightarrow
    g(v) = 1 \text{ where } g(x) = f(Lx+\ell).
\end{align}
So, any row operation on $H$ that leaves the first row intact
induces an affine transformation 
on the indicator polynomial ($f \to g$).

It is obvious that any affine transform $g$ of $f \in \RM(s,m)$ belongs again to $\RM(s,m)$.
Since affine transformations are composable and invertible,
they define an equivalence relation on $\RM(s,m)$.

\begin{cor}\label{cor:main}
The set of isomorphism classes of $r$-dimensional unital triorthogonal subspaces in $\FF_2^c$ is
in one-to-one correspondence with the affine equivalence classes of $ \RM(r-5,r-1)$ with Hamming weight $c$,
excluding those divisible by a polynomial of degree~$1$.
\end{cor}
\begin{proof}
We have characterized the indicator polynomial $f$ for a unital triorthogonal subspace:
$f$  has to be a polynomial of degree $\le r-5$.
We have to show that $f$ gives a generator matrix of rank $r$ 
if and only if it does not have a factor of degree~$1$.

Suppose $f = u v$ where $\deg u = 1$.
Then, there is an affine transformation on variables 
after which we have $u = x_1 + 1$.
This means that $x_1 f = 0$, 
implying that the second row of the associated generator matrix $H$ 
is zero.
Hence, $H$ has rank smaller than~$r$.

Conversely, if the rank of the generator matrix is less than~$r$,
then some row becomes zero after some row operation,
which means with certain affine transformation of variables we have 
$x_{r-1} f = 0 \in \FF_2[x_1,\ldots,x_{r-1}] / (x_i^2 - x_i)$,
which is only possible if $f$ has $x_{r-1} +1$ as a factor.
\end{proof}

We also note the following facts.
\begin{lem}\label{lem:bound}
Let $G$ be a triorthogonal matrix for a triorthogonal code with $d_Z \ge 2$.
Then, 
\begin{align}
    n\geq 2k.
\end{align}
\end{lem}
\begin{proof}
We may assume that no column of $G$ is zero,
and the assumption $d_Z \ge 2$ implies that no column of $G_0$ is zero.
The span $\calG_0$ of all rows of $G_0$ is supported on all $n$ components.
Since the average of the weight of all vectors in any binary vector space
is sum of the all averages of individual components,
the average weight of all vectors in $\calG_0$ is $n/2$.
Therefore, there exists a vector $\vec t$ of weight $\ge n/2$.
Let $\vec u = \vec 1_n - \vec t$, the indicator vector of zero components of $\vec t$.

For any $\vec v \in \FF_2^n$, let $\vec v \wedge \vec u$ denote the componentwise product of $\vec v$ and $\vec u$.
We may regard $\vec v \wedge \vec u$ as the restriction of $\vec v$ on the support of $\vec u$.
If $\vec g_1,\ldots, \vec g_k$ are rows of $G_1$,
then the vectors $\vec f_i = \vec g_i \wedge \vec u$ satisfy
\begin{align}
    \vec f_a \cdot \vec f_b = 
    |\vec g_a \wedge \vec g_b \wedge \vec 1| - |\vec g_a \wedge \vec g_b \wedge \vec t|
    =
    \begin{cases}
        |\vec g_a| - |\vec g_a \wedge \vec t| = 1 \bmod 2 & (a=b)\\
        |\vec g_a \wedge \vec g_b| - |\vec g_a \wedge \vec g_b \wedge \vec t| = 0 \bmod 2 & (a\neq b)
    \end{cases}.
\end{align}
Since every $\vec f_a$ is on at most $n/2$ bits,
the number $k$ of orthonormal and hence linearly independent vectors cannot exceed $n/2$.
\end{proof}
 
\begin{rem}\label{rem:dmax}
Let $\calH$ be a unital triorthogonal subspace. Define two functions of $k$:
\begin{align} 
d_{\text{max}}^{\text{even}}(k) &= \text{maximum $d_Z$ over all even descendants of $\calH$ with $k$ logical qubits.} \nonumber\\
d_{\text{max}}^{\text{odd}}(k) &= \text{maximum $d_Z$ over all odd descendants of $\calH$ with $k$ logical qubits.} \label{eq:maxdist}
\end{align}
Then, each of them is a nonincreasing function of $k$.
\end{rem}
\begin{proof}
Suppose $G_1$ and $G_0$ are the collections of odd and even weight rows, respectively, 
of a triorthogonal matrix with $k$ rows in $G_1$.
Increasing $k$ amounts to choosing a column,
permuting this column to the front,
and putting the matrix in a row echelon form such that
\begin{align}
    G = \left[\begin{array}{c}
     G_1  \\ \hline
     G_0
\end{array}   \right] 
\cong
\left[\begin{array}{c|c}
      0&G_1'  \\\hline
     1& \vec g' \\
     0& G_0' 
\end{array}  \right].
\end{align}
Suppose $\vec z$ be a minimum weight row vector that corresponds to a nontrivial $Z$ logical operator of the triorthogonal code of $G$.
By definition, $\vec z$ is orthogonal to all rows of $G_0$
but is not orthogonal to some rows of $G_1$.
Let $z^1$ be the first component of $\vec z$ 
and $\vec z'$ the rest, so that $\vec z = (z^1, \vec z')$.
Now, $\vec z'$ is orthogonal to all rows of $G_0'$
and $\vec z'$ is not orthogonal to some rows of $G_1'$.
Hence, $\vec z'$ corresponds to a nontrivial $Z$ logical operator of the new code encoding $k+1$ logical qubits.
If $z^1 = 0$, then the $Z$ distance of the new code is at most the old one.
If $z^1 = 1$, then the $Z$ distance of the new code is at most one less than the old one.
\end{proof}

\section{Reed--Muller codes of small weight}

We have so far shown that every triorthogonal code can be regarded 
as a descendant of a unial triorthogonal subspace, 
and that all unital triorthogonal subspaces are in one-to-one correspondence --- in the sense of~\cref{eq:indicator} --- 
with the affine equivalence classes of polynomials in $ \RM(r-5,r-1)$ (\cref{cor:main}).

Given an indicator polynomial $f$ of 
an $r$-dimensional unital triorthogonal subspace,
we know that $|f|$ is at least the minimum distance $2^{r-s-1}$ of $\RM(s,r-1)$ where $s = \deg f \le r - 5$.
In particular, $|f| \ge 16$.

\subsection{Codes with \texorpdfstring{$n+k\leq 30$}{}}

Kasami and Tokura~\cite{kasami1970weight} show that
every $f \in \RM(s,m)$ with $2^{m-s} \le |f| < 2^{m-s+1}$ is affine equivalent to one of the following polynomials:
\begin{align}
 &x_1 \cdots x_{s-q}(x_{s-q+1}\cdots x_s + x_{s+1}\cdots x_{s+q}), \quad \text{for } m\geq s+q \text{ and }s\geq q\geq 3,\nonumber \\
&x_1\cdots x_{s-2}(x_{s-1}x_s + x_{s+1}x_{s+2} + \cdots  + x_{s+2q-3}x_{s+2q-2})\quad \text{for }m-s+2\geq 2q\geq 2
\end{align}
Since our indicator polynomial $f$ should not be divisible by a linear factor,
the leading factor $x_1 \cdots x_{s-q}$ or $x_1\cdots x_{s-2}$ must be absent.
This implies that, in the first case, $s = q = \deg f \le m/2$,
or in the second case, $s = 2 = \deg f$.
Requiring $|f| \le 30$, we see that $m \le 8$ in the first case and $m \le 6$ in the second case.
Therefore, the following $5$ polynomials represent all affine equivalence classes
of $\RM(r-5,r-1)$ with no linear factor and of Hamming weight $\leq 30$.
 
\begin{enumerate}
    \item $p(x_1,x_2,x_3,x_4)= 1$ which corresponds to an $r = 5$ dimensional \utri subspace in $\FF_2^{16}$. This simple polynomial has the original [[15,1,3]] magic state distillation code by Bravyi and Kitaev~\cite{bravyi2005universal}, and first code of Bravyi--Haah family~\cite{bravyi2012magic} with parameters [[14,2,2]] as its even descendants. See the first row of~\cref{tab:2}.    
    \item $p(x_1,x_2,x_3,x_4,x_5,x_6)= x_1x_2 + x_3x_4$ which corresponds to a $7$-dimensional \utri subspace in $\FF_2^{24}$. This subspace contains the second code of Bravyi--Haah family~\cite{bravyi2012magic} as one of its descendants.
    \item $p(x_1,x_2,x_3,x_4,x_5,x_6)= x_1x_2 + x_3x_4 + x_5x_6$ which corresponds to a $7$-dimensional \utri subspace in $\FF_2^{28}$.
    \item $p(x_1,x_2,x_3,x_4,x_5,x_6,x_7)= x_1x_2x_3 + x_4x_5x_6$ which corresponds to an $8$-dimensional \utri subspace in $\FF_2^{28}$.
    \item $p(x_1,x_2,x_3,x_4,x_5,x_6,x_7,x_8)= x_1x_2x_3x_4 + x_5x_6x_7x_8$ which corresponds to a $9$-dimensional \utri subspace in $\FF_2^{30}$.
\end{enumerate}
With this list of polynomials, we can generate all triorthogonal subspaces in $\FF_2^c$, with $c\leq 30$ and compute the distance of all of its descendants. 
These results are reported in the first 5 rows of~\cref{tab:2}.

We observe in~\cref{tab:2} that in order for a triorthogonal code to have $d_Z \geq 2$, 
it must hold that $r-k=\rank(G_0)$ be at least $3$ for odd $r$ and at least $4$ for even $r$.
We prove this observation in the following,
which is a strengthening of a result in~\cite{bravyi2012magic}.
\begin{lem}
Let $G=\left[\begin{array}{c}
     G_1  \\\hline
     G_0
\end{array}   \right]$ be a triorthogonal matrix with the associated $Z$ distance $\geq 2$.
If $G$ is a descendant of an $r$-dimensional \utri subspace, then, 
\begin{equation}
    \rank (G_0) \geq \left \{\begin{array}{cc}
        4 & \quad\text{if }r\text{ is even} \\
        3 & \quad\text{if }r\text{ is odd}
    \end{array}\right..
\end{equation}
\end{lem}
\begin{proof}
The assumption that $d_Z \ge 2$ implies that $G_0$ does not contain any zero column.
Now, let $H$ be a generator matrix for the \utri subspace that has $G$ as its even descendant;
we will treat odd descendants later.
A \utri matrix can be constructed by padding $G$ with $I_k$, 
and we change the basis of $H$ such that the first row is the all-one vector.
\[
H = \left[\begin{array}{c|c|c}
         \multicolumn{3}{c}{\vec 1_{c}}    \\\hline
     0 & I_{k-1}& \star \\\hline
     0 & 0  &G_0
\end{array}  \right].
\]
Let $p$ be the indicator polynomial of $H$.
Since no column of $G_0$ is zero, we have
\begin{align*} 
&p(x_1,\cdots,x_{k-1},0,0,\cdots ,0)=1 \\
\Leftrightarrow &
    (x_1,\cdots,x_{k-1}) \text{ is one of the first $k$ columns of $H$}\\
\Leftrightarrow &
    (x_1,\cdots,x_{k-1}) \in \{\vec 0,  e_1,\cdots, e_{k-1}  \}.
\end{align*}
It follows that the polynomial $p(x_1,\cdots,x_{k-1},0,0,\cdots ,0)$ is the sum of $k$ ``delta functions'' 
\begin{align*} 
    p(x_1,\cdots,x_{k-1},0,\cdots ,0)
    &=\bar{x}_1\bar{x}_2 \cdots \bar{x}_{k-1}  + \sum_{j=1}^{k-1} \bar{x}_1\cdots\bar x_{j-1} x_j \bar x_{j+1} \cdots \bar{x}_{k-1},
\end{align*}
where $\bar x_i := (x_i+1)$.
Therefore,
\begin{align}
    \deg (p(x_1,\cdots,x_{k-1},0,0,\cdots ,0)) = 
    \begin{cases}
        k-2 & \text{if }k\text{ is even} \\
        k-1 & \text{if }k\text{ is odd}
    \end{cases} .
\end{align}
But $\deg (p(x_1,\cdots,x_{k-1},0,0,\cdots ,0)) \leq \deg p$ obviously,
and $\deg p \le r-5$ since $H$ is unital triorthogonal.
Hence,
\begin{align}
    \rank (G_0) = r-k  \geq  \begin{cases}
        3 & \quad\text{if }k\text{ is even} \\
        4 & \quad\text{if }k\text{ is odd}
\end{cases} .
\end{align}
This completes the proof if $G$ is an even descendant.

If $G$ is an odd descendant, the dimension of the parent triorthogonal space is $r = 1 + k + \rank(G_0)$,
but we have
\[
H = \left[\begin{array}{c|c|c}
         \multicolumn{3}{c}{\vec 1_{c}}    \\\hline
     0 & I_{k}& \star \\\hline
     0 & 0  &G_0
\end{array}  \right].
\]
So, a similar argument proves the lemma.
\end{proof}

\subsection{Codes with \texorpdfstring{$30<n+k\leq 38$}{}}

In the regime where $30 < n+k \le 38$, 
we have to examine indicator polynomials $f\in \RM(r-5,r-1)$ of unital triorthogonal spaces
that have weight $< 40$.
The upper bound $40$ is equal to $\frac 5 2 d$ where $d=16$ is the minimum distance of $\RM(r-5,r-1)$.
Since we have covered in the previous subsection the case where $|f| < 2 d = 32$,
here we assume that $|f| \ge 2d = 32$.
Let us use $m = r-1$ for the number of variables, 
so our indicator polynomial $f$ is always in $\RM(m-4,m)$.

Kasami, Tokura, and Azumi show~\cite[Thm.2]{kasami1976weight} that
if all of the following four conditions are satisfied for a polynomial function $f$ over $\FF_2^{m}$, 
(i) $f$ has no linear factor, 
(ii) $2d \le |f| < \frac 5 2 d$,
(iii) $\deg f \ge 4$,
and 
(iv) $m \ge 9$,
then
$f$ is affine equivalent to one and only one of polynomials in Table~1 of \cite{kasami1976weight}.
The condition (i) is true for us since we do not want redundant rows in a generating matrix for a unital triorthogonal space.
The condition (ii) is true for us because we are restricting our scope.
The conditions (iii) and (iv) may or may not be true,
and this is the place we will use computer search.
Still, we can shrink the search space by the following argument.

We know that the degree of $f$ must be $\le m-4$ to be an indicator polynomial of a unital triorthogonal space.
Here, $m$ is at least the number of distinct variables that appear in an expression of~$f$ but can be larger.
Let us show that 
\emph{
in our classification scope where $|f| = n+k < 40$, 
it suffices to consider cases where either $m = 4 + \deg f$ or $f = f(x_1,x_2,x_3,x_4,x_5) = 1$.}
\begin{proof}
If $m \ge 6 + \deg f$, then $f \in \RM(m-6,m)$ and the weight of~$f$ is at least~$64$.
If $m = 5 + \deg f$, then $|f| \ge 32$.
In this case, if $32 < |f| < 2 \cdot 32$, then \cite[Lem.1.(2)]{kasami1976weight} implies that $|f| \ge 64 - 16 = 48$,
which is beyond our scope, and we may assume $|f| = 32$.
But, then, \cite[Lem.1.(2.1)]{kasami1976weight} implies that 
$f=1$ if it does not have a linear factor.
\end{proof}
Hence, if $m \ge 9$, we may assume that $\deg f \ge 5$ and the condition (iii) and (iv) are satisfied.
Examining \cite[Table~1]{kasami1976weight},
we find that there are only five polynomials that satisfy the condition $\deg f = m-4$.
They are reported in \cref{tab:2}, Polynomials 15, 30, 36, 37, and 38.
The case of~$f(x_1,x_2,x_3,x_4,x_5) = 1$ corresponds to Polynomial 6 in \cref{tab:2}.

Now, the remaining cases of our classification problem are when $4 < 4 + \deg f = m \le 8$.
We are going to classify polynomials $p \in \RM(4,8)$ where $\deg p = 4$ and $|p| < 40$,
but where $p$ is allowed to have a linear factor.
This is sufficient because any $f \in \RM(m-4,m)$ with $\deg f = m-4 < 4$ with $|f| < 40$ 
may be multiplied by $x_{m+1}x_{m+2}\cdots x_{8}$ to become degree $4$ and still $|x_{m+1}x_{m+2}\cdots x_{8} f| < 40$.
After finding all such polynomials $p$, 
we can remove any linear factors and recover the cases with $6$ or $7$ variables.

\cite[Theorem~1.(1)]{kasami1976weight} says that for a polynomial $p$ of degree $4$ (or larger),
if $p$ has weight less than $40$,
then $p$ is affine equivalent to a polynomial of form
\begin{align}
    p(x_1,x_2,\cdots,x_8) = x_7 g(x_1,x_2,\cdots, x_6) + x_8 h(x_1,x_2,\cdots, x_6) + x_7 x_8 u(x_1,x_2,\cdots, x_6)
\end{align}
where $\deg g\leq 3$, $\deg h \leq 3$, and $\deg u\leq 2$, or more succinctly, $g,h \in \RM(3,6), u \in \RM(2,6)$.
It is easy to see that $|p|=|g|+|h|+|g+h+u|$ by setting $x_7,x_8 = 0,1$.
Using simple change of variables (e.g., $x_7 \to x_7+x_8$ or $x_8 \to x_7+x_8$)
one can assume without loss of generality that 
\begin{align}
    |g| \leq |h| \leq |g+h+u|. \label{eq:order}
\end{align}
We call $g$ and $h$ the \emph{base} polynomials for $p$. 

Given a pair of base polynomials, 
there are only $\binom{6}{2} + \binom{6}{1} + \binom{6}{0} = 22$ monomials of degree~$2$, $1$, or~$0$ 
that can be included in the polynomial $u$. In our computer search, 
we construct all $2^{22}$ possible polynomials $u$, 
and check their Hamming weights to see if they match a target weight.
This algorithm usually finds many affine equivalent polynomials. 
Therefore, for each polynomial $p$ with a target weight, 
we perform a heuristic optimization over the affine equivalence class of $p$ 
with the goal of minimizing the number of monomials. 
In this way, we find a much smaller number of distinct polynomials, 
most of them having only a few monomials. 
See~\cite{github} for more details. 
It only remains to run the search on all possible combinations of basis polynomials $g,h\in \RM(3,6)$. 

Since $|p|=|g|+|h|+|g+h+u|\leq 38$, 
one can easily see that $|g|, |h|\leq 18$ as a consequence of~\cref{eq:order}. 
In order to see what polynomials qualify as a base pair,
we first solve a more tractable problem of classifying affine equivalence classes of low weight polynomials in $\RM(3,6)$.
The search space is much smaller in this case and with the aid of \cite[Theorem~1.(1)]{kasami1976weight} 
we can perform a full computer search and classification. 
See~\cref{tab:m6} for the results.

  \begin{table}[t]
      \centering
      \begin{tabular}{|c|l|}
      \hline
          Weight & Representatives of affine equivalence classes of $ \RM(3,6)$ \\\hline\hline
          $8$ & \tabitem $x_1x_2x_3$
                   \\\hline
          $12$ & \tabitem $x_1(x_2x_3 + x_4x_5)$
                   \\\hline
          $14$ & \tabitem $x_1x_2x_3 + x_4 x_5x_6$
                   \\\hline
          $16$ & 
         \tabitem $x_1x_2$\\&
        \tabitem $x_1(x_2 + x_3x_4)$\\&
        \tabitem $x_1(x_2 + x_3x_4 + x_5x_6)$\\&
        \tabitem $(x_1+1)x_2x_3 + x_1x_4x_5$\\&
        \tabitem $x_2x_3x_4 + x_1x_3x_5 + x_1x_2x_6$
                   \\\hline
          $18$ & 
         \tabitem $x_1x_2 + x_2x_3x_5 + x_1x_4x_6$\\&
        \tabitem $x_1x_2x_3 + x_2x_3x_4 + x_1x_2x_5 + x_1x_3x_6 + x_4x_5x_6$\\\hline
           \end{tabular}
      \caption{
      List of representative polynomials of affine equivalence classes of $ \RM(3,6)$ with weight less than or equal to $18$. 
      This list has been constructed using a computer search 
      over the polynomials of the form $x_1 g(x_3,x_4,x_5,x_6)+x_2 h(x_3,x_4,x_5,x_6)+x_1x_2 u(x_3,x_4,x_5,x_6)$ 
      with $\deg g,\deg h\leq 2$, and $\deg u\leq 1$. 
      We know from \cite[Thm.~1]{kasami1976weight} that all elements of $\RM(3,6)$ are affine equivalent 
      to a polynomial of the above form.
      } 
      \label{tab:m6}
  \end{table}

Coming back to the problem of classifying elements of $ \RM(4,8)$, 
we list all possible combinations of base pairs based on their weights, 
up to affine transformations on variables $x_1,x_2,\cdots,x_6$: 
\begin{enumerate}
    \item $|g| =0, |h| =0$, which means that both~$g$ and~$h$ are identically zero. 
    \item $|g| =0, |h| =8$, where $g=0$ and $h$ can be chosen to be $x_1 x_2 x_3$. This is because every weight $8$ element of $ \RM(3,6)$ is affine equivalent to $x_1x_2x_3$; see~\cref{tab:m6}.
    \item $|g| =0, |h| =12$, where $g=0$ and $h$ can be chosen to be $x_1(x_2x_3 + x_4x_5)$; see~\cref{tab:m6}. 
    \item $|g| =0, |h| =14$, where $g=0$ and $h$ is chosen to be $x_1x_2x_3 + x_4 x_5x_6$; see~\cref{tab:m6}.
    \item $|g| =0, |h| =16$, where $g=0$ and $h$ is one of the $5$ polynomials with weight~$16$ in~\cref{tab:m6}.
    \item $|g| =0, |h| =18$, where $g=0$ and $h$ is one of the $2$ polynomials with weight~$18$ in~\cref{tab:m6}. 
    Since $|p| = |h|+|h+u|$, and $|h+u|\geq|h|$, we should only consider these base polynomials for~$p$ with $|p|\geq 36$. 

    \item $|g| =8, |h| =8$. 
    In this case we can set $g=x_1x_2x_3$. 
    We only know that $h$ is affine equivalent to $x_1x_2x_3$;
    we cannot immediately set $h$ to be $x_1x_2x_3$ because we may not be able to bring both $g$ and $h$ to their canonical affine representatives simultaneously.
    However, one can see by further investigation that using affine transformations that fix $g=x_1x_2x_3$, 
    one can bring $h$ to one of at most $32$ options~%
    \footnote{
        If $h = (x_1+1)(x_2+1)(x_3+1)$, 
        the support of~$h$ is a 3-dimensional subspace $\{(0,0,0,x_4,x_5,x_6)~:~x_4,x_5,x_6 = 0,1 \}$.
        Since $g$ is affine equivalent to $h$, the support of $g$ is an affine subspace of dimension~3.
        If both affine subspaces contain the origin,
        the only invariant of the pair of the subspaces under linear (not general affine) transformations,
        is the dimension of their intersection, which can be~$0$, $1$, $2$, or~$3$.
        Given an intersection dimension, the support of~$g$ can be translated along $3$ directions normal to the support of~$h$.
        This gives $4 \cdot 2^3 = 32$ options for~$g$.
        This is an overcounting;
        if the intersection dimension is~$0$, then the support of~$g$ is $\{(x_1,x_2,x_3,0,0,0)~:~x_1,x_2,x_3 = 0,1\}$,
        but any translation of this does not change the polynomial~$h$.
    }.
    We use all of these $32$ pairs as our basis polynomials.
    
    \item $|g| =8, |h| =12$. 
    Similar to the previous case, we can set $h = x_1(x_2x_3 + x_4x_5)$. 
    Then by affine transformations that fix $h$, 
    we can bring $g$ (which itself is affine equivalent to $x_1x_2x_3$) to one of $224$ choices.
    (It might be possible to reduce this number of possibilities).
    We implement the computer search for all of these $224$ basis polynomials.
    
    \item $|g| =8, |h| =14$.
    We can set $h=x_1x_2x_3 + x_4 x_5x_6$ (see~\cref{tab:m6}), 
    and consider $h$-preserving affine transformations.
    In this case, we can find $264$ candidates for $g$.
    These basis pairs are only relevant for polynomials with $|p| = |g|+|h|+|g+h+u| \geq 8+14+14=36$.
    
    \item $|g| =12, |h| =12$.
    This case implies that $|p| = |g|+|h|+|g+h+u| \geq 12+12+12=36$.
    We set $h = x_1(x_2x_3 + x_4x_5)$, the second row of~\cref{tab:m6}.
    Using $h$-preserving affine transformations, we find $1404928$ options for $g$; 
    this is likely an overcounting.
    We use all of the $1404928$ possible polynomials as base pairs in our computer search. 
    This search is the most costly, as we have to examine more than $5 \times 10^{12}$ polynomials, 
    accounting for $2^{22}$ options for the polynomial $u$ given a base pair. 
    See our classification code~\cite{github}.
\end{enumerate}

In this way, we conclude the classification of all \utri subspaces embedded in $\FF^c_2$ with $c< 40$. 
All of the polynomials found in this search, in addition to the distances of their even descendants, 
are reported in~\cref{tab:2} and~\cref{fig:classification}. 
See~\cite{github} for explicit matrices. 
We do not report $d_{\text{max}}^{\text{odd}}(k)$ explicitly 
as we numerically observed that $d_{\text{max}}^{\text{odd}}(k) =d_{\text{max}}^{\text{even}}(k+1)$ for all polynomials. 

Our results show that there is no \tri code with distance higher than $3$ among the codes with $n+k\leq 38$. 
However, we find the smallest code with $k=3$ and $d=3$ which has $n=35$ physical qubits.
It is an even descendant of code number~33 in~\cref{tab:2} with the generator matrix~\cref{eq:gen35}.

\section{Divisibility at level 3}

A subspace $\calH \subseteq \FF_2^c$ is said to be \emph{divisible at level 3}~\cite{haah2018towers}
if there exists a vector $\vec t \in \ZZ_8^c$ with all odd entries
such that $\vec h \cdot \vec t = 0 \bmod 8$ for all $\vec h \in \calH$.
Every level~3 divisible subspace is triorthogonal,
but the converse was not known to be true.
We observe that the converse is \emph{false};
in \cref{tab:2}, Code~3, 17, 20, 23, 28, and 33 are not divisible at level~3.

To determine that those codes are not level~3 divisible but all other codes are level~3 divisible,
we develop an efficient algorithm as follows.
The level~3 divisibility is equivalent~\cite[Lem.III.2]{haah2018towers} to the following set of conditions
on a basis $\{\vec h_1,\ldots, \vec h_r \}$ of the subspace $\calH$:
\begin{enumerate}
    \item[0.] $\vec t^i = 1 \bmod 2$ for all $i$,
    \item[1.] $\vec h_a \cdot \vec t = 0 \bmod 8$ for all $a$,
    \item[2.] $(\vec h_a \wedge \vec h_b) \cdot \vec t = 0 \bmod 4$ for all $a \le b$, and
    \item[3.] $(\vec h_a \wedge \vec h_b \wedge \vec h_c) \cdot \vec t = 0 \bmod 2$ for all $a \le b \le c$.
\end{enumerate}
Given a triorthogonal subspace $\calH$, we know that the third condition is satisfied.
We work with an $r$-by-$c$ matrix $M$ in the reduced row echelon form whose rows form a basis for $\calH$,
where we assume that the left $r$-by-$r$ submatrix is the identity matrix.
Let $N$ be $\binom{r}{2}$-by-$(c-r)$ matrix 
where each row, indexed by a pair $(a,b)$ with $1 \le a < b \le r$,
is $\vec h_a \wedge \vec h_b$.
Condition~0 and~2 are that $\vec t|_{c-r}$, the restriction of $\vec t$ on the last $c-r$ entries,
should satisfy 
\begin{align}
N \cdot \vec t|_{c-r} = 0 \bmod 4 , \qquad \vec t|_{c-r}^i = 1 \bmod 2 ~\forall i.\label{eq:Nt}
\end{align}
Conversely, suppose we have $\vec t|_{c-r}$ that fulfills \eqref{eq:Nt}.
Then, we can easily find a full $\vec t$ such that $M \cdot \vec t = 0 \bmod 8$ as follows.
Let $\vec t|_r$ denote the first $r$ entries of $\vec t$.
For each row $\vec h_a$ of $M$, we have to solve an equation $\vec h_a|_r \cdot \vec t|_r = - \vec h_a|_{c-r} \cdot \vec t|_{c-r} \bmod 8$.
Since $\vec h_a|_r$ has sole nonzero entry $1$ at the $a$-th position,
this equation clearly has a solution, but we need to check if Condition~0 is fulfilled.
Since Condition~1 is fulfilled mod 2, the subvector $\vec h_a|_{c-r}$ contains an odd number of ones,
and hence $\vec h_a|_{c-r} \cdot \vec t|_{c-r}$ is odd, and therefore Condition~0 can be fulfilled.

Hence, given a triorthogonal subspace,
a desired $\vec t$ exists if and only if there is a solution to~\cref{eq:Nt}.
We know that \cref{eq:Nt} has a solution over $\FF_2$;
the triorthogonality implies that the all-one vector is a solution.
Thus, any solution over $\ZZ_4$ can only differ from this all-one vector~$\vec 1$
by some vector with even entries.
That is, we may write $\vec t|_{c-r} = \vec 1 + 2 \vec v \mod 4$ where $\vec v$ is a binary vector of length $c-r$.

Let $U$ be an integer matrix such that $U N \bmod 2$ is in the reduced row echelon form.
Since $UN \cdot \vec 1 = 0 \bmod 2$,
we know $UN \cdot \vec 1 \bmod 4$ consists of even entries.
If $UN \bmod 4$ has a row $\rho$ of all even entries,
then $\rho \cdot (\vec 1 + 2\vec v) = \rho \cdot \vec 1 \bmod 4$ for any binary vector $\vec v$.
Hence, \cref{eq:Nt} has a solution only if $\rho \cdot \vec 1 = 0 \bmod 4$ for any all-even row~$\rho$ of $UN$.
Conversely, if $\rho \cdot \vec 1 = 0 \bmod 4$ for any all-even row $\rho$ of $UN$,
then it is straightforward to find a vector $\vec v$ such that $UN \cdot (\vec 1 + 2\vec v) = 0 \bmod 4$
since for every row of $UN$ which is nonzero over $\FF_2$ 
has an odd entry such that any other entry in its column is even.

In summary, an efficient algorithm  
to test if a subspace~$\cal H \subseteq \FF_2^c$ is divisible at level $3$
is as follows.
(i) take a matrix $M$ in the reduced row echelon form over $\FF_2$ whose row span is $\calH$,
(ii) permute columns of $M$ such that the left block is the $r$-by-$r$ identity matrix,
(iii) make a binary $\binom{r}{2}$-by-$(c-r)$ matrix $N$ by enumerating entrywise products but ignoring first $r$ entries of all pairs of rows of $M$,
(iv) compute the reduced row echelon form $UN$ of $N$ over $\FF_2$,
and
(v) test the integer matrix $UN$ for each all-even row of $UN$ whether the sum of all entries of the row is zero mod 4.

\section{Conclusion}

Triorthogonal codes are a versatile class of codes for constructing magic state distillation protocols. 
In particular, they are the most general CSS codes for $T$ state distillation, 
and furthermore considering non-CSS stabilizer codes does not seem to improve codes parameters~\cite{rengaswamy2020optimality}.
We have shown that it suffices to consider unital triorthogonal codes,
characterized by the property that the parent triorthogonal space contains all-one vector.
By indexing triorthogonal codes by Reed--Muller codewords,
we have classified all unital triorthogonal codes with $k$ logical qubits on $n$ physical qubits where $n+k \le 38$.
Our classification reveals new instances such as $[[35,3,3]]$ and $[[28,2,3]]$ codes.
In addition, we have shown the limitations of \tri codes with small parameters, 
for example, that there is no code with $Z$ distance larger than~$3$ when $n+k\leq 38$ 
and that the first three Bravyi--Haah codes are extremal in this regime.

Although a triorthogonal code as an abstract CSS code is not necessarily 
tied to a magic state distillation circuit,
it serves as a basic ingredient for many distillation circuits~\cite{campbell2017unifying,haah2018codes,Litinski2019}.
Given an enveloping design of distillation protocols,
our main result~\cref{tab:2} can be used for 
selecting the most appropriate instance.

\begin{center}
\begin{longtable}{|c|l|c|c||c|c|c|c|c|c|c|c|}
\hline        
         \rule{0pt}{5ex} \multirow{2}{*}{ $\#$} &    \multirow{2}{*}{ $p(x_1,x_2,\cdots,x_{r-1})$ }& \multirow{2}{*}{ $\,\,\,r\,\,\,$} &     \multirow{2}{*}{$\,\,\,c\,\,\,$}  & \multicolumn{6}{c}{$\quad\quad d_{\text{max}}^{\text{even}}(k)$\rotatebox[origin=c]{-90}{$\quad\quad$} }&\\ 
         &&&&\rotatebox[origin=c]{-90}{$k=1\,$}&\rotatebox[origin=c]{-90}{$k=2\,$}&\rotatebox[origin=c]{-90}{$k=3\,$}&\rotatebox[origin=c]{-90}{$k=4\,$}&\rotatebox[origin=c]{-90}{$k=5\,$}&\rotatebox[origin=c]{-90}{$k=6\,$}&\rotatebox[origin=c]{-90}{$k=7\,$}\\\hline\hline
\endfirsthead
\multicolumn{4}{c}%
{\tablename\ \thetable\ -- \textit{Continued from previous page}} \\
\hline        
         \rule{0pt}{5ex} \multirow{2}{*}{$\#$} &      \multirow{2}{*}{ $p(x_1,x_2,\cdots,x_{r-1})$ }& \multirow{2}{*}{ $\,\,\,r\,\,\,$} &     \multirow{2}{*}{$\,\,\,c\,\,\,$}  & \multicolumn{6}{c}{$\quad\quad d_{\text{max}}^{\text{even}}(k)$\rotatebox[origin=c]{-90}{$\quad\quad$} }&\\ 
         &&&&\rotatebox[origin=c]{-90}{$k=1\,$}&\rotatebox[origin=c]{-90}{$k=2\,$}&\rotatebox[origin=c]{-90}{$k=3\,$}&\rotatebox[origin=c]{-90}{$k=4\,$}&\rotatebox[origin=c]{-90}{$k=5\,$}&\rotatebox[origin=c]{-90}{$k=6\,$}&\rotatebox[origin=c]{-90}{$k=7\,$}\\\hline\hline
\endhead
\hline \multicolumn{4}{r}{\textit{Continued on next page}} \\
\endfoot
\hline
\endlastfoot

 $1$&$1 $&$\,\,\,5\,\,\,$&$\,\,16\,\,$&$\,\,\,3\,\,\,$&$\,\,\,2\,\,\,$&$\,\,\,1\,\,\,$&$\,\,\,1\,\,\,$&$\,\,\,1\,\,\,$&$\,\,\,-\,\,\,$&$\,\,\,-\,\,\,$\\\hline
$2$&$x_{1 }x_{2} + x_{3 }x_{4} $&$\,\,\,7\,\,\,$&$\,\,24\,\,$&$\,\,\,3\,\,\,$&$\,\,\,2\,\,\,$&$\,\,\,2\,\,\,$&$\,\,\,2\,\,\,$&$\,\,\,1\,\,\,$&$\,\,\,1\,\,\,$&$\,\,\,1\,\,\,$\\\hline
$3$&$x_{1 }x_{2} + x_{3 }x_{4} + x_{5 }x_{6} $&$\,\,\,7\,\,\,$&$\,\,28\,\,$&$\,\,\,3\,\,\,$&$\,\,\,2\,\,\,$&$\,\,\,2\,\,\,$&$\,\,\,2\,\,\,$&$\,\,\,1\,\,\,$&$\,\,\,1\,\,\,$&$\,\,\,1\,\,\,$\\\hline
$4$&$x_{1 }x_{2 }x_{3} + x_{4 }x_{5 }x_{6} $&$\,\,\,8\,\,\,$&$\,\,28\,\,$&$\,\,\,3\,\,\,$&$\,\,\,2\,\,\,$&$\,\,\,2\,\,\,$&$\,\,\,2\,\,\,$&$\,\,\,1\,\,\,$&$\,\,\,1\,\,\,$&$\,\,\,1\,\,\,$\\\hline
$5$&$x_{1 }x_{2 }x_{3 }x_{4} + x_{5 }x_{6 }x_{7 }x_{8} $&$\,\,\,9\,\,\,$&$\,\,30\,\,$&$\,\,\,3\,\,\,$&$\,\,\,3\,\,\,$&$\,\,\,2\,\,\,$&$\,\,\,2\,\,\,$&$\,\,\,1\,\,\,$&$\,\,\,1\,\,\,$&$\,\,\,1\,\,\,$\\\hline
$6$&$1 $&$\,\,\,6\,\,\,$&$\,\,32\,\,$&$\,\,\,3\,\,\,$&$\,\,\,2\,\,\,$&$\,\,\,1\,\,\,$&$\,\,\,1\,\,\,$&$\,\,\,1\,\,\,$&$\,\,\,1\,\,\,$&$\,\,\,-\,\,\,$\\\hline
$7$&$x_{1 }x_{2} + x_{3} $&$\,\,\,7\,\,\,$&$\,\,32\,\,$&$\,\,\,3\,\,\,$&$\,\,\,2\,\,\,$&$\,\,\,2\,\,\,$&$\,\,\,2\,\,\,$&$\,\,\,1\,\,\,$&$\,\,\,1\,\,\,$&$\,\,\,1\,\,\,$\\\hline
$8$&$x_{2 }x_{3} + x_{1 }x_{4} + x_{5} $&$\,\,\,7\,\,\,$&$\,\,32\,\,$&$\,\,\,3\,\,\,$&$\,\,\,2\,\,\,$&$\,\,\,2\,\,\,$&$\,\,\,2\,\,\,$&$\,\,\,1\,\,\,$&$\,\,\,1\,\,\,$&$\,\,\,1\,\,\,$\\\hline
$9$&$(x_{1} + 1) x_{2 }x_{3} + x_{1 }x_{4 }x_{5} $&$\,\,\,8\,\,\,$&$\,\,32\,\,$&$\,\,\,3\,\,\,$&$\,\,\,2\,\,\,$&$\,\,\,2\,\,\,$&$\,\,\,2\,\,\,$&$\,\,\,1\,\,\,$&$\,\,\,1\,\,\,$&$\,\,\,1\,\,\,$\\\hline
$10$&$x_{1 }x_{3 }x_{4} + x_{1 }x_{2 }x_{5} + x_{2 }x_{3 }x_{6} $&$\,\,\,8\,\,\,$&$\,\,32\,\,$&$\,\,\,3\,\,\,$&$\,\,\,2\,\,\,$&$\,\,\,2\,\,\,$&$\,\,\,2\,\,\,$&$\,\,\,1\,\,\,$&$\,\,\,1\,\,\,$&$\,\,\,1\,\,\,$\\\hline
$11$&$x_{1 }x_{2 }x_{4} + x_{1 }x_{5 }x_{6} + x_{2 }x_{3 }x_{7} $&$\,\,\,8\,\,\,$&$\,\,32\,\,$&$\,\,\,3\,\,\,$&$\,\,\,2\,\,\,$&$\,\,\,2\,\,\,$&$\,\,\,2\,\,\,$&$\,\,\,1\,\,\,$&$\,\,\,1\,\,\,$&$\,\,\,1\,\,\,$\\\hline
$12$&$x_{1 }x_{2 }x_{3 }x_{4} + x_{1 }x_{2 }x_{3} + (x_{1 }x_{2} + x_{3 }x_{4}) x_{5 }x_{6} $&$\,\,\,9\,\,\,$&$\,\,32\,\,$&$\,\,\,3\,\,\,$&$\,\,\,2\,\,\,$&$\,\,\,2\,\,\,$&$\,\,\,2\,\,\,$&$\,\,\,2\,\,\,$&$\,\,\,2\,\,\,$&$\,\,\,1\,\,\,$\\\hline
$13$&$x_{1 }x_{2 }x_{3 }x_{4} + x_{2 }x_{3 }x_{4 }x_{5} + x_{1 }x_{5 }x_{6 }x_{7} $&$\,\,\,9\,\,\,$&$\,\,32\,\,$&$\,\,\,3\,\,\,$&$\,\,\,2\,\,\,$&$\,\,\,2\,\,\,$&$\,\,\,2\,\,\,$&$\,\,\,1\,\,\,$&$\,\,\,1\,\,\,$&$\,\,\,1\,\,\,$\\\hline
$14$&$x_{1 }x_{2 }x_{3 }x_{5} + x_{3 }x_{4 }x_{5 }x_{7} + (x_{1 }x_{2 }x_{4} + x_{1 }x_{2}) x_{6} $&$\,\,\,9\,\,\,$&$\,\,32\,\,$&$\,\,\,3\,\,\,$&$\,\,\,2\,\,\,$&$\,\,\,2\,\,\,$&$\,\,\,2\,\,\,$&$\,\,\,2\,\,\,$&$\,\,\,2\,\,\,$&$\,\,\,1\,\,\,$\\\hline
$15$&$x_{1 }x_{2 }x_{3 }x_{4 }x_{5} + (x_{1} + 1) x_{6 }x_{7 }x_{8 }x_{9} $&$\,\,\,10\,\,\,$&$\,\,32\,\,$&$\,\,\,3\,\,\,$&$\,\,\,3\,\,\,$&$\,\,\,2\,\,\,$&$\,\,\,2\,\,\,$&$\,\,\,1\,\,\,$&$\,\,\,1\,\,\,$&$\,\,\,1\,\,\,$\\\hline
$16$&$x_{1 }x_{2 }x_{3 }x_{4} + (x_{1 }x_{2} + x_{5 }x_{6}) x_{7 }x_{8} $&$\,\,\,9\,\,\,$&$\,\,34\,\,$&$\,\,\,3\,\,\,$&$\,\,\,3\,\,\,$&$\,\,\,2\,\,\,$&$\,\,\,2\,\,\,$&$\,\,\,2\,\,\,$&$\,\,\,2\,\,\,$&$\,\,\,1\,\,\,$\\\hline
$17$&$x_{1 }x_{3} + x_{4 }x_{5} + x_{2 }x_{6} + 1 $&$\,\,\,7\,\,\,$&$\,\,36\,\,$&$\,\,\,3\,\,\,$&$\,\,\,2\,\,\,$&$\,\,\,2\,\,\,$&$\,\,\,2\,\,\,$&$\,\,\,1\,\,\,$&$\,\,\,1\,\,\,$&$\,\,\,1\,\,\,$\\\hline
$18$&$x_{2 }x_{3 }x_{5} + x_{1 }x_{4 }x_{6} + x_{1 }x_{2} $&$\,\,\,8\,\,\,$&$\,\,36\,\,$&$\,\,\,3\,\,\,$&$\,\,\,2\,\,\,$&$\,\,\,2\,\,\,$&$\,\,\,2\,\,\,$&$\,\,\,1\,\,\,$&$\,\,\,1\,\,\,$&$\,\,\,1\,\,\,$\\\hline
$19$&$x_{2 }x_{3 }x_{5} + x_{1 }x_{6 }x_{7} + (x_{1 }x_{2} + x_{1 }x_{3}) x_{4} $&$\,\,\,8\,\,\,$&$\,\,36\,\,$&$\,\,\,3\,\,\,$&$\,\,\,2\,\,\,$&$\,\,\,2\,\,\,$&$\,\,\,2\,\,\,$&$\,\,\,1\,\,\,$&$\,\,\,1\,\,\,$&$\,\,\,1\,\,\,$\\\hline
$20$&$x_{2 }x_{3 }x_{4} + x_{1 }x_{3 }x_{5} + x_{1 }x_{2 }x_{6} + x_{1 }x_{4 }x_{7} $&$\,\,\,8\,\,\,$&$\,\,36\,\,$&$\,\,\,3\,\,\,$&$\,\,\,2\,\,\,$&$\,\,\,2\,\,\,$&$\,\,\,2\,\,\,$&$\,\,\,1\,\,\,$&$\,\,\,1\,\,\,$&$\,\,\,1\,\,\,$\\\hline
$21$&$x_{3 }x_{4 }x_{6} + x_{1 }x_{2 }x_{7} + (x_{2 }x_{3} + x_{1 }x_{4}) x_{5} $&$\,\,\,8\,\,\,$&$\,\,36\,\,$&$\,\,\,3\,\,\,$&$\,\,\,2\,\,\,$&$\,\,\,2\,\,\,$&$\,\,\,2\,\,\,$&$\,\,\,1\,\,\,$&$\,\,\,1\,\,\,$&$\,\,\,1\,\,\,$\\\hline
$22$&$x_{1 }x_{2 }x_{3} + x_{2 }x_{3 }x_{4} + x_{1 }x_{2 }x_{5} + x_{1 }x_{2} + (x_{1 }x_{3} + x_{4 }x_{5}) x_{6} $&$\,\,\,8\,\,\,$&$\,\,36\,\,$&$\,\,\,3\,\,\,$&$\,\,\,2\,\,\,$&$\,\,\,2\,\,\,$&$\,\,\,2\,\,\,$&$\,\,\,1\,\,\,$&$\,\,\,1\,\,\,$&$\,\,\,1\,\,\,$\\\hline
\multirow{2}{*}{$\,23\,$}&$(x_{1} + x_{2}) x_{5 }x_{6} + (x_{1 }x_{2} + x_{1 }x_{3}) x_{4} $&\multirow{2}{*}{$\,\,\,8\,\,\,$}&\multirow{2}{*}{$\,\,36\,\,$}&\multirow{2}{*}{$\,\,\,3\,\,\,$}&\multirow{2}{*}{$\,\,\,2\,\,\,$}&\multirow{2}{*}{$\,\,\,2\,\,\,$}&\multirow{2}{*}{$\,\,\,2\,\,\,$}&\multirow{2}{*}{$\,\,\,1\,\,\,$}&\multirow{2}{*}{$\,\,\,1\,\,\,$}&\multirow{2}{*}{$\,\,\,1\,\,\,$}\\
&$\qquad\qquad\qquad+ (x_{2 }x_{3} + x_{3 }x_{5} + x_{4 }x_{6}) x_{7} $&&&&&&&&&\\\hline
$24$&$x_{1 }x_{2 }x_{3 }x_{4} + x_{1 }x_{2 }x_{5 }x_{6} + x_{3 }x_{4 }x_{5 }x_{7} $&$\,\,\,9\,\,\,$&$\,\,36\,\,$&$\,\,\,3\,\,\,$&$\,\,\,2\,\,\,$&$\,\,\,2\,\,\,$&$\,\,\,2\,\,\,$&$\,\,\,2\,\,\,$&$\,\,\,2\,\,\,$&$\,\,\,1\,\,\,$\\\hline
$25$&$x_{2 }x_{3 }x_{4 }x_{5} + x_{1 }x_{2 }x_{4 }x_{7} + x_{1 }x_{3 }x_{6 }x_{8} $&$\,\,\,9\,\,\,$&$\,\,36\,\,$&$\,\,\,3\,\,\,$&$\,\,\,2\,\,\,$&$\,\,\,2\,\,\,$&$\,\,\,2\,\,\,$&$\,\,\,2\,\,\,$&$\,\,\,2\,\,\,$&$\,\,\,1\,\,\,$\\\hline
$26$&$x_{1 }x_{2 }x_{4 }x_{5} + x_{1 }x_{2 }x_{5 }x_{6} + x_{1 }x_{3 }x_{4 }x_{7} + x_{2 }x_{3 }x_{6 }x_{8} $&$\,\,\,9\,\,\,$&$\,\,36\,\,$&$\,\,\,3\,\,\,$&$\,\,\,2\,\,\,$&$\,\,\,2\,\,\,$&$\,\,\,2\,\,\,$&$\,\,\,2\,\,\,$&$\,\,\,2\,\,\,$&$\,\,\,1\,\,\,$\\\hline
$27$&$x_{1 }x_{2 }x_{3 }x_{4} + x_{1 }x_{3 }x_{5 }x_{6} + (x_{1 }x_{2 }x_{5} + x_{2 }x_{4 }x_{6}) x_{7} $&$\,\,\,9\,\,\,$&$\,\,36\,\,$&$\,\,\,3\,\,\,$&$\,\,\,2\,\,\,$&$\,\,\,2\,\,\,$&$\,\,\,2\,\,\,$&$\,\,\,2\,\,\,$&$\,\,\,2\,\,\,$&$\,\,\,1\,\,\,$\\\hline
$28$&$x_{2 }x_{3 }x_{4 }x_{5} + (x_{1} + x_{2}) x_{4 }x_{5 }x_{6} + x_{1 }x_{2 }x_{3 }x_{7} + x_{1 }x_{3 }x_{6 }x_{8} $&$\,\,\,9\,\,\,$&$\,\,36\,\,$&$\,\,\,3\,\,\,$&$\,\,\,2\,\,\,$&$\,\,\,2\,\,\,$&$\,\,\,2\,\,\,$&$\,\,\,2\,\,\,$&$\,\,\,2\,\,\,$&$\,\,\,1\,\,\,$\\\hline
$29$&$x_{1 }x_{2 }x_{4 }x_{6} + x_{1 }x_{2 }x_{3 }x_{7} + (x_{4 }x_{5 }x_{6} + x_{3 }x_{5 }x_{7}) x_{8} $&$\,\,\,9\,\,\,$&$\,\,36\,\,$&$\,\,\,3\,\,\,$&$\,\,\,2\,\,\,$&$\,\,\,2\,\,\,$&$\,\,\,2\,\,\,$&$\,\,\,2\,\,\,$&$\,\,\,2\,\,\,$&$\,\,\,1\,\,\,$\\\hline
$30$&$x_{1 }x_{2 }x_{3 }x_{4 }x_{5} + x_{3 }x_{4 }x_{6 }x_{7 }x_{8} + (x_{5} + 1) x_{6 }x_{7 }x_{8 }x_{9} $&$\,\,\,10\,\,\,$&$\,\,36\,\,$&$\,\,\,3\,\,\,$&$\,\,\,3\,\,\,$&$\,\,\,2\,\,\,$&$\,\,\,2\,\,\,$&$\,\,\,2\,\,\,$&$\,\,\,2\,\,\,$&$\,\,\,1\,\,\,$\\\hline
$31$&$x_{1 }x_{2 }x_{3 }x_{4} + x_{2 }x_{3 }x_{7 }x_{8} + (x_{1 }x_{2} + x_{1 }x_{4}) x_{5 }x_{6} $&$\,\,\,9\,\,\,$&$\,\,38\,\,$&$\,\,\,3\,\,\,$&$\,\,\,3\,\,\,$&$\,\,\,2\,\,\,$&$\,\,\,2\,\,\,$&$\,\,\,2\,\,\,$&$\,\,\,2\,\,\,$&$\,\,\,1\,\,\,$\\\hline
$32$&$x_{1 }x_{2 }x_{3 }x_{5} + x_{1 }x_{2 }x_{4 }x_{6} + x_{2 }x_{4 }x_{5 }x_{7} + x_{1 }x_{3 }x_{6 }x_{8} $&$\,\,\,9\,\,\,$&$\,\,38\,\,$&$\,\,\,3\,\,\,$&$\,\,\,3\,\,\,$&$\,\,\,2\,\,\,$&$\,\,\,2\,\,\,$&$\,\,\,2\,\,\,$&$\,\,\,2\,\,\,$&$\,\,\,1\,\,\,$\\\hline
\multirow{2}{*}{$\,33\,$}&$x_{1 }x_{2 }x_{4 }x_{5} + x_{1 }x_{4 }x_{5 }x_{6} + (x_{2 }x_{3 }x_{4} + x_{1 }x_{5 }x_{6}) x_{7}  $&\multirow{2}{*}{$\,\,\,9\,\,\,$}&\multirow{2}{*}{$\,\,38\,\,$}&\multirow{2}{*}{$\,\,\,3\,\,\,$}&\multirow{2}{*}{$\,\,\,3\,\,\,$}&\multirow{2}{*}{$\,\,\,3\,\,\,$}&\multirow{2}{*}{$\,\,\,2\,\,\,$}&\multirow{2}{*}{$\,\,\,2\,\,\,$}&\multirow{2}{*}{$\,\,\,2\,\,\,$}&\multirow{2}{*}{$\,\,\,1\,\,\,$}\\
&$\qquad\qquad\qquad\qquad+ (x_{1 }x_{2 }x_{3} + x_{1 }x_{2 }x_{6} + x_{2 }x_{3 }x_{7}) x_{8}$&&&&&&&&&\\\hline
$34$&$x_{1 }x_{2 }x_{3 }x_{4} + x_{2 }x_{4 }x_{5 }x_{6} + x_{1 }x_{5 }x_{6 }x_{7} + x_{1 }x_{3 }x_{7 }x_{8} $&$\,\,\,9\,\,\,$&$\,\,38\,\,$&$\,\,\,3\,\,\,$&$\,\,\,3\,\,\,$&$\,\,\,2\,\,\,$&$\,\,\,2\,\,\,$&$\,\,\,2\,\,\,$&$\,\,\,2\,\,\,$&$\,\,\,1\,\,\,$\\\hline
\multirow{2}{*}{$\,35\,$}&$x_{1 }x_{2 }x_{3 }x_{4} + x_{1 }x_{2 }x_{3 }x_{5} + x_{1 }x_{5 }x_{7 }x_{8}  $&\multirow{2}{*}{$\,\,\,9\,\,\,$}&\multirow{2}{*}{$\,\,38\,\,$}&\multirow{2}{*}{$\,\,\,3\,\,\,$}&\multirow{2}{*}{$\,\,\,3\,\,\,$}&\multirow{2}{*}{$\,\,\,2\,\,\,$}&\multirow{2}{*}{$\,\,\,2\,\,\,$}&\multirow{2}{*}{$\,\,\,2\,\,\,$}&\multirow{2}{*}{$\,\,\,2\,\,\,$}&\multirow{2}{*}{$\,\,\,1\,\,\,$}\\
&$\qquad\qquad\qquad\qquad+(x_{2 }x_{3 }x_{4} + x_{1 }x_{4 }x_{5}) x_{6}$&&&&&&&&&\\\hline
$36$&$x_{1 }x_{2 }x_{3 }x_{4 }x_{5} + x_{4 }x_{5 }x_{6 }x_{7 }x_{8} + x_{3 }x_{6 }x_{7 }x_{8 }x_{9} $&$\,\,\,10\,\,\,$&$\,\,38\,\,$&$\,\,\,3\,\,\,$&$\,\,\,3\,\,\,$&$\,\,\,2\,\,\,$&$\,\,\,2\,\,\,$&$\,\,\,2\,\,\,$&$\,\,\,2\,\,\,$&$\,\,\,1\,\,\,$\\\hline
\multirow{2}{*}{$\,37\,$}&$x_{1 }x_{2 }x_{3 }x_{4 }x_{5} + x_{3 }x_{4 }x_{6 }x_{7 }x_{8}  $&\multirow{2}{*}{$\,\,\,10\,\,\,$}&\multirow{2}{*}{$\,\,38\,\,$}&\multirow{2}{*}{$\,\,\,3\,\,\,$}&\multirow{2}{*}{$\,\,\,3\,\,\,$}&\multirow{2}{*}{$\,\,\,2\,\,\,$}&\multirow{2}{*}{$\,\,\,2\,\,\,$}&\multirow{2}{*}{$\,\,\,2\,\,\,$}&\multirow{2}{*}{$\,\,\,2\,\,\,$}&\multirow{2}{*}{$\,\,\,1\,\,\,$}\\
&$\qquad\qquad\qquad\qquad+ (x_{2 }x_{5 }x_{6 }x_{7} + x_{2 }x_{6 }x_{7 }x_{8}) x_{9}$&&&&&&&&&\\\hline
$38$&$x_{1 }x_{2 }x_{3 }x_{4 }x_{5 }x_{6} + (x_{10 }x_{3 }x_{4} + x_{10 }x_{5 }x_{6} + x_{10 }x_{5}) x_{7 }x_{8 }x_{9} $&$\,\,\,11\,\,\,$&$\,\,38\,\,$&$\,\,\,3\,\,\,$&$\,\,\,3\,\,\,$&$\,\,\,2\,\,\,$&$\,\,\,2\,\,\,$&$\,\,\,2\,\,\,$&$\,\,\,2\,\,\,$&$\,\,\,1\,\,\,$\\\hline

\caption{
List of affine representatives of $\RM(r-5,r-1)$ with no linear factor and of weight $c < 40$. 
Every polynomial listed above is the indicator function (in the sense of~\cref{eq:indicator}) 
of an $r$-dimensional \utri subspace embedded in $\FF_2^c$, 
with $c:=|p|$ being the Hamming weight of the polynomial as a Reed--Muller codeword. 
Every \utri code with $n+k\leq 38$ and no repeating columns 
can be constructed as a descendant of one of these \utri subspaces. 
For each $k$, we have listed $d_{\text{max}}^{\text{even}}(k)$, 
the maximum distance achieved by the even descendants of a given \utri subspace with $k$ logical qubits. 
We have observed that  $d_{\text{max}}^{\text{odd}}(k) =d_{\text{max}}^{\text{even}}(k+1)$ for codes in the table, 
and therefore the distances of odd descendants are not listed.
The first three instances of Bravyi--Haah~\cite{bravyi2012magic} codes are even descendants of codes~1, 2, and~12. 
Except Polynomials~12 and~13, 
we have checked that all \utri subspaces corresponding to the above polynomials 
have different weight enumerator functions, and therefore they are affine inequivalent.
The inequivalence of polynomials~12 and~13
is evident from the fact that they have different $d_{\text{max}}^{\text{odd}}(k)$ functions.
}
\label{tab:2}
\end{longtable}
\end{center}

\appendix

\section{Delaying Clifford corrections in magic state distillation circuits}

The standard magic state injection circuit consists of a multiqubit measurement
followed by a conditional Clifford correction.
While Pauli operators can be implemented passively by Pauli frame tracking
so they do not require any feedback from a classical controller to a quantum hardware,
the Clifford correction must be implemented on the quantum hardware by decision from the classical controller.
The cost of this feedback will be escalated for a magic state factory
because many magic states are consumed in a factory.
In this appendix, we propose a circuit that reduces the classical feedback.
In a specific protocol below, our proposal will reduce the time cost, too.

\subsection{\texorpdfstring{$\ZZ_4$}{}-valued quadratic forms}\label{sec:Z4forms}

Let $M$ be a symmetric matrix over $\ZZ_4$.
By a quadratic form $q$ by $M$, we mean the function $q: z \mapsto z M z^T \in \ZZ_4$,
where $z$ is a vector over $\FF_2$~
\footnote{
    It may look weird to mix up $\FF_2$ and $\ZZ_4$,
    but this object appears in algebraic topology and 
    is a quadratic enhancement of a mod~2 intersection form.
    Since we are not enhancing anything here,
    we choose to give a simpler name, a $\ZZ_4$-valued quadratic form.
    All material in \cref{sec:Z4forms} is well known in math literature.
    See e.g. \cite{Taylor2008}.
}.
It is to be checked whether this function is well defined 
since the coefficient ring in the domain is a quotient ring of that in the codomain.
For any integral vectors $z$ and $y$,
we see that $(z + 2y) M (z+2y)^T = z M z^T + 2 y M z^T + 2 z M y^T + 4 y M y^T = z M z^T \mod 4$.
This shows that the function $q: \FF_2^m \to \ZZ_4$ is well defined by $M$.

Although every element of $M$ belongs to $\ZZ_4$,
the off-diagonal elements basically reside in $\FF_2$ for the following reason.
If a symmetric matrix $N$ over $\ZZ_4$ has zero diagonal and even off-diagonal entries,
then $z N z^T = \sum_{a,b} z_a N_{ab} z_b = 2 \sum_{a < b} z_a N_{ab} z_b = 0 \mod 4$.
Therefore, two $\ZZ_4$-valued quadratic forms by $M$ and $M + N$ are equal.
Hence, if we are given an equation $M = M' \mod 2$ of symmetric matrices,
the $\ZZ_4$-valued quadratic form by $M$ is determined by $M'$ up to diagonal elements.
Thus it makes sense to define the $\FF_2$-rank of $q$ by the $\FF_2$-rank of $M$.

\begin{lem}\label{prop:quadratic}
    For any $\ZZ_4$-valued quadratic form $q : \FF_2^m \to \ZZ_4$,
    there exists an $r$-by-$m$ matrix $W$ and an $m$-by-$m$ diagonal matrix $D$,
    both over $\FF_2$, such that $q$ is defined by $W^T W + 2\cdot D$
    and $r_0 \le r \le r_0 + 1$, $r_0 = \mathrm{rank}_{\FF_2}(q)$.
    Here, $2 \cdot : \FF_2 \to \ZZ_4$ is the standard additive group embedding.
\end{lem}
\begin{proof}
    Let $M$ be any matrix over $\ZZ_4$ by which $q$ is defined.
    By the discussion above, it suffices to find $W$ such that $W^T W = M \mod 2$
    because we can read off the diagonal $D$ from $M - W^T W \mod 4$.
    From now on, we work over~$\FF_2$.
    
    If $M$ has $1$ in the diagonal, 
    then there is an invertible matrix $E$ such that $E^{-T} M E^{-1}$ is diagonal.
    (This is well known~\cite{albert1938symmetric,lempel1975matrix,janusz2007parametrization,haah2017magic}.)
    But over $\FF_2$ any diagonal matrix is the identity matrix with some number of trailing zeros in the diagonal.
    Hence, we may write $M = E^T \begin{pmatrix} I_{r_0} & 0 \\ 0 & 0 \end{pmatrix} E$.
    Keeping only the first $r_0$ row of $E$, we find $W$ with $r_0$ rows.
    
    If the diagonal of $M$ is zero, then we consider one-larger matrix 
    $\begin{pmatrix} 1 & 0 \\ 0 & M \end{pmatrix}$ 
    and find $W'$ to reproduce it by the argument in the previous paragraph.
    The desired $W$ is obtained by removing the first column of $W'$.
    The number of rows of $W$ is $r_0+1$.
\end{proof}
A variant of Gauss elimination can be used to find $W$ given $M$,
putting the computational complexity to $\mathcal O(m^3)$.

\subsection{Diagonal Clifford Gates}

The math below is essentially contained in~\cite[Chap.~2]{OMeara},
which can be understood if one is familiar with the correspondence 
between Clifford groups and symplectic groups~\cite{DM}.
However, we were not able to identify an exact claim in~\cite{OMeara,DM}
that gives our result.
We choose to be explicit here.

Let $Z = \ket 0 \bra 0 - \ket 1 \bra 1$ be the standard Pauli $Z$ matrix.
For any binary vector $v = (v_1,\ldots,v_m) \in \FF_2^m$,
let $Z(v)$ be the tensor product
\begin{align}
    Z(v) = \bigotimes_{j=1}^m Z^{v_j}
\end{align}
of $Z$ and the identity matrices.
By diagonal Clifford gates,
we mean any product 
\begin{align}
S(V) = \prod_{v \in V} \exp\left( \frac{\i\pi}{4} - \frac{\i\pi}{4} Z(v) \right) \label{eq:SV}
\end{align}
where $v$ ranges over some set $V \subseteq \FF_2^m$.
Since $Z$ is diagonal, any two diagonal Clifford gates commute with each other,
so the product is unambiguous.
The set of all diagonal Clifford gates includes the usual $S$ gate and the control-$Z$ gate.    
Note that
\begin{align}
    S(V)^2 = \prod_v \exp\left(\frac{\i \pi}{2}  - \frac{\i \pi}{2} Z(v) \right) = \prod_{v \in V} Z(v)
\end{align}
is a Pauli operator.

Let us examine the action of $S(V)$ more explicitly.
The following formula will be useful;
\begin{align}
    x \bmod 2 = x^2 \bmod 4
\end{align}
which is true for any integer $x$ if ``mod 2'' and ``mod 4'' 
are interpreted as the nonnegative smallest integer remainder 
after division by~$2$ and~$4$, respectively.
From now on, we identify $\FF_2$ as a subset of $\ZZ$.
By abuse of notation, let $V = (V_{ab})$ be the matrix over $\FF_2$
whose rows are vectors in the set $V \subseteq \FF_2^m$.
On an arbitrary computational basis state $\ket{z} = \ket{z_1,\ldots,z_m}$ where $z_j \in \FF_2$
we have
\begin{align}
    S(V)\ket{z}
    &=
    \prod_{a} \exp\left( \frac{\i\pi}{4} - \frac{\i\pi}{4} \prod_b (-1)^{V_{ab}z_b} \right) \ket{z}\\
    &=
    \ket{z} \prod_a \exp\left[ \frac{\i \pi}{2} \left( \sum_b V_{ab}z_b \bmod 2 \right)\right] \nonumber\\
    &= \ket z \exp\left[ \sum_a \frac{\i \pi}{2} 
    \left( \sum_b V_{ab}z_b  \right)^2 \right] \nonumber\\
    &= \ket z \exp\left[ \frac{\i \pi}{2} 
     \sum_{a,b,c} V_{ab} z_b V_{ac} z_c  \right] \nonumber\\
    &= \ket z \exp\left[ \frac{\i \pi}{2} z V^T V z^T  \right] \nonumber
\end{align}
where in the last line we used a vector-matrix notation in which $z$ is a row vector.
The expression in the exponential implies that the action of $S(V)$ is determined by
a symmetric matrix
\begin{align}
    M = V^T V
\end{align}
viewed as a function $q:\FF_2^m \to \ZZ_4$ by $q(z) = z M z^T$.

By \cref{prop:quadratic}, any such function $q$ 
can be realized by a matrix $W$ with $m+1$ or less rows and a diagonal matrix $D$ as
\begin{align}
    q(z) = z (W^T W) z^T + 2\cdot z D z^T
    \end{align}
The diagonal matrix $2\cdot D$ corresponds to $S(D)^2$ which is a tensor product of Pauli matrix $Z$.
This means that
\begin{align}
    S(V) = S(W) Z(\diag(D)).
\end{align}
Observe that $S(W)$ consists of at most $m+1$ rotation gates
whereas $S(V)$ contains $|V|$ rotation gates which can be exponentially large in $m$.

\subsection{One level higher}

An analogous calculation can be done for diagonal, Clifford-conjugated $T$ gates.
We first observe that
\begin{align}
    x \bmod 2 = 2x^3 + x^2 - 2x \bmod 8 \label{eq:mod8}
\end{align}
for any integer $x$ where ``mod 2'' and ``mod 8'' are interpreted as taking the smallest nonnegative remainder after division. 
Now, we define 
\begin{align}
T(V) = \prod_{v \in V} \exp\left( \frac{\i\pi}{8} - \frac{\i\pi}{8} Z(v) \right) 
\end{align}
where $v$ ranges over some set $V \subseteq \FF_2^m$.
Then, for any computational basis state $\ket z$ we have
\begin{align}
    &T(V)\ket{z}\nonumber\\
    &=
    \prod_{a} \exp\left( \frac{\i\pi}{8} - \frac{\i\pi}{8} \prod_b (-1)^{V_{ab}z_b} \right) \ket{z}\nonumber\\
    &=
    \ket{z} \prod_a \exp\left[ \frac{\i \pi}{4} \left( \sum_b V_{ab}z_b \bmod 2 \right)\right] \nonumber\\
    &=
    \ket{z} \prod_a \exp\left[ \frac{\i \pi}{4} \left( 
                2\sum_{b,c,d} V_{ab}z_b V_{ac}z_c V_{ad}z_d 
                +\sum_{b,c} V_{ab}z_b V_{ac}z_c
                -2\sum_{b} V_{ab}z_b
                \right)\right] \qquad \text{(using \cref{eq:mod8})}\nonumber\\
    &=
    \ket{z} \prod_a \exp\left[ \frac{\i \pi}{4} \left( 
                12\sum_{b<c<d} V_{ab}z_b V_{ac}z_c V_{ad}z_d 
                -2\sum_{b<c} V_{ab}z_b V_{ac}z_c
                +\sum_{b} V_{ab}z_b
                \right)\right] \nonumber\\
    &=
    \ket{z} \exp\left[ \i\pi  
                \left(\sum_{b<c<d} \sum_a V_{ab} V_{ac} V_{ad} z_b z_c z_d \right)
                -\frac{\i \pi}{2}\left(\sum_{b<c} \sum_a V_{ab} V_{ac} z_b z_c \right)
                +\frac{\i \pi}{4} \left(\sum_{b} \sum_a V_{ab} z_b \right)
                \right]. \label{eq:TV}
\end{align}
Note that the cubic term is a collection of $CCZ$ gates,
the quadratic term is a collection of $CS$ gate, 
and the linear term is a collection of $T$ gates.

\subsection{Application to \texorpdfstring{$T$}{}-distillation circuits}

There are many circuit implementations possible for a given triorthogonal code,
and here we consider a ``space-efficient'' one.
The general idea of this space-efficient protocol was sketched in~\cite[\S II.D]{haah2018codes},
and more explicitly written in~\cite{Litinski2019}.
The specific protocol we consider here is the following.
Let $G$ be a triorthogonal matrix with even weight rows form $G_0$ and odd weight rows $G_1$.

\begin{enumerate}
\item Prepare $\ket + ^{\otimes (k+g_0)}$ where $g_1=k, g_0$ are the numbers of rows in $G_1$ and $G_0$, respectively.
\item For each $\ell = 1,2,\ldots,n$, apply $\exp(-\i\pi \bar Z_\ell /8)$ 
where $\bar Z_\ell = Z((G_{j\ell})_j)$ is the tensor product of $Z$ for each $1$ in \emph{column} $\ell$ of $G$.
\item After a diagonal Clifford correction $S[G]$ (that is vacuous if $G$ descends from a triply even subspace; see below),
measure single-qubit $X$ on the last $g_0$ qubits corresponding to the rows of $G_0$.
\item Postselect on all $+1$ outcomes.
\item The magic states are in the first $k$ qubits corresponding to the rows of $G_1$.
\end{enumerate}
If $V$ is the collection of all \emph{columns} of $G$,
then $T(V)$ implements $T$ gates on the qubits that corresponds to the rows of $G_1$ up to a diagonal Clifford.
Indeed,
the cubic term of \cref{eq:TV} vanishes due to the triorthogonality of $G$.
The gate that implements
\begin{align}
S[G] \ket z = 
\ket{z} \exp\left[
                -\frac{\i \pi}{2}\left(\sum_{b<c} \sum_a G_{ab} G_{ac} z_b z_c \right)
                +\frac{\i \pi}{4} \left(\left(\sum_{a,b} G_{ab} z_b \right) - \left(\sum_{a,b} G_{ab} z_b \bmod 2\right)\right) 
                \right].\label{eq:SG}
\end{align}
is a diagonal Clifford.

The rotation $e^{-\i \pi \bar Z_\ell / 8}$ may not be an elementary operation,
in which case it must be induced by a $T$ injection~\cite{knill2004,bravyi2005universal}.
The $T$ injection is achieved by the following measurement sequence.
\begin{enumerate}
    \item Prepare an ancilla qubit in $T$ state $\ket 0 + e^{\i\pi/4} \ket 1$ 
    with possible noise. This can be provided by an earlier round of $T$ distillation.
    \item Measure $Z_\text{ancilla} \otimes \bar Z_\ell$ to obtain an outcome $t_\ell = \pm 1$.
    \item Measure $X_\text{ancilla}$. If the outcome is $-1$, apply $\bar Z_\ell$.
    \item If $t_\ell = -1$, apply $\exp(-\i \pi \bar Z_\ell / 4)$.
\end{enumerate}
Except for the last Clifford correction $e^{-\i\pi \bar Z_\ell / 4}$,
every measurement is a multiqubit Pauli measurement,
which can be implemented by lattice surgery techniques~\cite{Horsman2012}.
Note that the ancilla in $T$ injection is measured out before the last step,
so we may reuse it.
The last Clifford correction upon $t_\ell = -1$ can be implemented in a number of ways,
but it is essentially an $S$ injection~\cite{Litinski2019}.
Recall that in a lattice surgery architecture,
the application of a Pauli operator is always passive
and does not correspond to any action on a quantum device;
one keeps track of what Pauli frame a qubit is in
and interpret any measurement outcome in the Pauli frame.
Perhaps importantly, a measurement-outcome dependent Pauli operator
does not require any classical feedback.

In the absence of any noise, the Clifford correction is needed with probability~$\tfrac 1 2$,
and the total number of Clifford corrections in the above $T$~distillation protocol
follows a binomial distribution $B(n,\tfrac 1 2)$.
In the presence of some noise in the circuit, 
the measurement outcome distribution can be biased,
but in the regime of practical interests the bias is small.
Such a stochastic process is less favorable than a fully deterministic process
because it makes it harder to synchronize operations across the quantum device.
Moreover, the Clifford corrections depend on classical feedback
where we have to know $n$ bits where~$n$ is the number of columns of~$G$.
This might slow down the execution of the overall distillation protocol.

We propose to delay all the Clifford corrections until all input $T$ states are consumed.
This is possible since any operation on the data qubits, that corresponds to rows of $G$, 
is diagonal in the $Z$ basis, and so does the $S$ correction.
That is, we just collect all the outcomes $t_\ell$ for $\ell =1,2,\ldots,n$,
and we apply
\begin{align}
    \prod_{\ell: t_\ell = -1} \exp(-\i \pi \bar Z_\ell / 4).
\end{align}
This is in the form of $S(V)$ in \cref{eq:SV} where
\begin{align}
    V = \{ v \in \FF_2^{k+g_0} ~|~ \exists \ell : t_\ell = -1, ~v = \text{$\ell$-th column of }G\} .
\end{align}
Hence, according to \cref{prop:quadratic}, $S(V)$ can be implemented by at most $k+g_0+1$ $S$-injections.
In all triorthogonal matrices we know, $k + g_0 + 1 < n$.
For example, in $15$-to-$1$ protocol~\cite{bravyi2005universal}, $k + g_0 + 1= 6 < 15 = n$.
In $116$-to-$12$ protocol~\cite{haah2018codes}, $k + g_0 + 1= 30 < 116 = n$.

Let us flesh out the protocol. 
As before, $G$ denotes a triorthogonal matrix.
\begin{enumerate}
    \item Prepare $\ket + ^{\otimes (k+g_0)}$ where $g_1=k, g_0$ are the numbers of rows in $G_1$ and $G_0$, respectively.
    \item For each $\ell = 1,2,\ldots,n$, do the following.
    
    (i) Prepare an ancilla qubit in $T$ state $\ket 0 + e^{\i\pi/4} \ket 1$ 
    with possible noise. This can be provided by an earlier round of $T$ distillation.

    (ii) Measure $Z_\text{ancilla} \otimes \bar Z_\ell$ to obtain an outcome $t_\ell = \pm 1$.
    Here, $\bar Z_\ell = Z((G_{j\ell})_j)$ is the tensor product of $Z$ for each $1$ in \emph{column} $\ell$ of $G$.

    (iii) Measure $X_\text{ancilla}$. If the outcome is $-1$, apply $\bar Z_\ell$.

    \item Let $C$ be the collection of indices $\ell$ such that $t_\ell = -1$.
    Let $H$ be the submatrix of $G$ by choosing columns of indices in $C$,
    and let $V$ consist of all the columns of $H$ \emph{and} a set of vectors that implement $S[G]$ of \cref{eq:SG}.
    Find matrices $W$ and $D$ such that \cref{eq:SV} holds.

    \item Apply $S(W) Z(\diag(D))$.
    
    \item Measure individual $X$ on the last $g_0$ qubits corresponding to the rows of $G_0$.
    Postselect on all $+1$ outcomes.
    
    \item The magic states are in the first $k$ qubits corresponding to the rows of $G_1$.
\end{enumerate}

In the proposed protocol,
the stochastic nature of the process is not entirely eliminated,
but the number of Clifford correction is now upper bounded by $k+g_0 +1$, a smaller number than~$n$,
and classical feedback is required only once, rather than $n$ times in sequence,
in between all $T$ consumption (Step~2) and $S(V)$ application (Step~4).

If the triorthogonal code allows $T$ and $T^\dagger$ gates to induce logical $T$ gates without any further Clifford correction
(empty ``$S[G]$'')~\cite{haah2018towers},
then some Clifford correction $e^{-\i\pi \bar Z_\ell / 4}$ is called on $t_\ell = +1$ rather than $t_\ell = -1$.
Our proposal can be used in that case, too, by letting $V$ to be the collection of all needed $S$-corrections.

\ssection{Acknowledgments}
SN is supported by the Walter Burke Institute for Theoretical Physics and IQIM at Caltech. 
Part of this work was done while SN was an intern
in the Quantum Architectures and Computation group (QuArC), Microsoft Research.

\bibliography{references}

\end{document}